\documentclass[10pt]{amsart}

\usepackage[utf8]{inputenc}
\usepackage{amsfonts}
\usepackage{amssymb}
\usepackage{amsmath,amsthm}
\usepackage{amscd}
\usepackage{graphics}
\usepackage{cancel}
\usepackage{pdfsync}
\usepackage[british]{babel} 
\usepackage{mathrsfs}
\usepackage{enumitem}
\usepackage{stmaryrd}
\usepackage{mathrsfs}
\usepackage{wasysym}
\usepackage{latexsym}
\usepackage{mathtools}
\usepackage{afterpage}
\usepackage{threeparttable}

\usepackage{empheq,etoolbox}
\usepackage{xcolor}
\usepackage[colorlinks,linkcolor=blue,citecolor=blue,urlcolor=blue]{hyperref}

\numberwithin{equation}{section}

\newcommand{\beq}{\begin{equation}}
\newcommand{\eeq}{\end{equation}}
\newcommand{\bea}{\begin{eqnarray}}
\newcommand{\eea}{\end{eqnarray}}
\newcommand{\nn}{\nonumber}
\newcommand\noi{\noindent}
\newcommand{\tbf}{\textbf}

\newcommand{\rd}{\mathrm{d}}

\newcommand{\bk}{\begin{cases}}
\newcommand{\ek}{\end{cases}}
\newcommand{\bepm}{\begin{pmatrix} }
\newcommand{\epm}{\end{pmatrix}}
\newcommand{\bs}{\boldsymbol}

\newcommand{\f}{\frac}

\newtheorem{definition}{Definition}
\newtheorem{proposition}{Proposition}
\newtheorem{theorem}{Theorem}
\newtheorem{corollary}{Corollary}
\newtheorem{lemma}{Lemma}
\theoremstyle{definition}

\newtheorem{remark}{\textbf{Remark}}

\usepackage{enumerate}
\usepackage{float}
\usepackage{cancel}
\usepackage{hyperref}
\usepackage{mathtools}
\usepackage{afterpage}

\allowdisplaybreaks

\begin{document}

\title[Partial separability and $\omega \mathscr{H}$ manifolds]{Partial separability and symplectic-Haantjes manifolds}

\author{Daniel Reyes}
\address{Departamento de F\'{\i}sica Te\'{o}rica, Facultad de Ciencias F\'{\i}sicas, Universidad
Complutense de Madrid, 28040 -- Madrid, Spain \\ and Instituto de Ciencias Matem\'aticas, C/ Nicol\'as Cabrera, No 13--15, 28049 Madrid, Spain}
\email{danreyes@ucm.es, daniel.reyes@icmat.es}
\author{Piergiulio Tempesta}
\address{Departamento de F\'{\i}sica Te\'{o}rica, Facultad de Ciencias F\'{\i}sicas, Universidad
Complutense de Madrid, 28040 -- Madrid, Spain \\  and Instituto de Ciencias Matem\'aticas, C/ Nicol\'as Cabrera, No 13--15, 28049 Madrid, Spain}
\email{piergiulio.tempesta@icmat.es, ptempest@ucm.es}
\author{Giorgio Tondo}
\address{Dipartimento di Matematica, Informatica e Geoscienze, Universit\`a  degli Studi di Trieste,
piaz.le Europa 1, I--34127 Trieste, Italy (retired)}
\email{tondo@units.it}

\date{8th July, 2024}

\begin{abstract}
A theory of partial separability for classical Hamiltonian systems is proposed in the context of Haantjes geometry.

As a general result, we show that the knowledge of a non-semisimple symplectic-Haantjes manifold for a given Hamiltonian system is sufficient to construct sets of coordinates (called Darboux-Haantjes coordinates) that allow both the partial separability of the associated Hamilton-Jacobi equations and the block-diagonalization of the operators of the corresponding Haantjes algebra. 

We also introduce a novel class of Hamiltonian systems, characterized by the existence of a generalized St\"ackel matrix, which by construction are partially separable. They widely generalize the known families of partially separable Hamiltonian systems. The new systems can be described in terms of semisimple but non-maximal-rank symplectic-Haantjes manifolds. 
\end{abstract}

\maketitle

\tableofcontents

\section{Introduction: separable systems}

\subsection{Background} One of the most fundamental problems of classical mechanics is the construction of separating variables for the Hamilton-Jacobi (HJ) equations associated with a Hamiltonian integrable system. This subject has a long history and still represents a very active research area. Starting from the seminal works by Liouville, Jacobi, St\"{a}ckel, Levi-Civita, Eisenhart, Arnold, etc., the problem of separability has been formulated within increasingly sophisticated theoretical frameworks, which in turn have provided us with a novel insight into the geometry of integrable systems.

As is well known, given a Hamiltonian function $H: T^*Q\rightarrow \mathbb{R}$ (where $Q$ is the configuration manifold), a complete integral of its stationary HJ equations is defined as a family of solutions
\beq 
W=W(q^1,\ldots, q^n; h_1,\ldots,h_n)
\eeq
depending smoothly on the parameters $(h_1,\ldots,h_n)$, and satisfying the condition 
\beq \label{eq:CHJ}
\det \bigg[ \frac{\partial^2 W}{\partial q^i\partial h_j} \bigg]\neq 0 \ .
\eeq 
The function $W$ is also said to be a Jacobi characteristic function. 
From a geometric point of view, $W$ is the generating function of a Lagrangian foliation of $T^*Q$, both transversal to the fibers and compatible with $H$. Such a foliation is  explicitly defined by the $n$ parametric equations for the momenta

\begin{equation} \label{eq:pLF}
p_i=\frac{\partial W}{ \partial q^i}(\boldsymbol{q}; h_1,\ldots,h_n), \qquad\qquad i=1,\ldots,n \ ,
\end{equation}
where $\bs{q}=(q^{1},\ldots,q^{n})$. Thanks to conditions \eqref{eq:CHJ}, such equations can be solved w.r.t. the constants $(h_1,\ldots,h_n)$, yielding: 
\begin{equation} \label{eq:Hi}
h_i=H_i(\boldsymbol{q},\boldsymbol{p}),
\end{equation}
with $\bs{p}=(p_{1},\ldots,p_{n})$. These relations provide us with $n$ independent integrals of motion in involution both with each other and with $H$. In particular, these integrals are  \textit{vertically independent}, namely, they satisfy the condition
\begin{equation} \label{eq:Hiv}
\det\left[\frac{\partial{H_i}}{\partial p_k}\right] \neq 0 \ .
\end{equation}

 Thus, a complete integral $W$ for $H$ determines a completely integrable Hamiltonian system $(H_1, \ldots, H_n)$ in the sense of Liouville and  Arnold and a common solution of the stationary HJ equations associated with the Hamiltonians of the system. Conversely, a family of independent functions $(H_1, \ldots, H_n)$ in involution and vertically independent determines a complete integral of the HJ equations, obtained by solving Eqs. \eqref{eq:Hi}  w.r.t. the momenta $(p_1,\ldots,p_n)$.
\par
A Hamiltonian system $H$  is said to be \textit{separable} if there exists a system of canonical coordinates $(q^1,\ldots, q^n,p_1,\ldots,p_n)$ and a solution of the HJ equations that takes the additively separated form \begin{equation} \label{eq:W}
W= \sum_{k=1}^n W_k(q^k; h_1,\ldots,h_n) =\sum_{k=1}^n \int^{q^k} {\frac{\partial W_k}{\partial \xi^k}(\xi^k; h_1, \ldots, h_n)d\xi^k}.
\end{equation}
Such a condition is equivalent to the existence of 
$n$ independent smooth functions $(H_1,\ldots, H_n)$ - on which $H$ is functionally dependent, and $n$ equations, called the  Jacobi-Sklyanin separation equations (SE) for $(H_1,\ldots, H_n)$, having the form
\begin{equation} \label{eq:Sk}
\phi_i(q^i,p_i; H_1,\ldots,H_n)=0 \, , \qquad    {\rm  det}\left[\frac{\partial \phi_i}{\partial H_j }\right] \neq 0 \, , \quad i=1,\ldots,n \, .
 \end{equation} 
When $H_1 = h_1, \ldots, H_n=h_n$, the SE represent implicitly the  Lagrangian transversal foliation \eqref{eq:pLF}.  Moreover, the SE allow us to construct a separated solution to the HJ equations. 
In fact, by solving Eqs. (\ref {eq:Sk})
with respect to the momenta $p_k$, one gets the Lagrangian foliation in the form  of a set of separated first-order ODEs:
\begin{equation}
p_k=\frac{\partial W_k}{\partial q^k}(q^k; h_1,\ldots,h_n) \, , \qquad k=1,\ldots,n \, ,
\end{equation}
which implies that $W$ acquires the additively separated expression \eqref{eq:W}. Finally, when the SE are \textit{affine} in the functions  $(H_1,\ldots, H_n)$ and the Hamiltonian function $H=H_1$ is quadratic and orthogonal in the momenta 
\beq \label{eq:1.9}
H=\frac{1}{2} \sum_{i=1}^n g^{ii}(\boldsymbol{q}) p_i^2 +V(\boldsymbol{q}) \, ,
\eeq
 one recovers
the standard theory of separability due to St\"ackel \cite{StCRAS1893,Pere}. We recall that St\"ackel's approach provides us with a characterization of \textit{totally  separable systems} by means of an invertible \textit{St\"ackel matrix} $\bs{S}$ and a \textit{St\"ackel vector} $|\bs{U}\rangle$, of the form
\beq \label{eq:SM}
\bs{S}= 
\begin{bmatrix}
S_{11}(q^1)& \cdots & S_{1n}(q^1) \\ 
\vdots  &\ddots & \vdots \\
S_{n1}(q^n) & \cdots & S_{nn}(q^n)
\end{bmatrix} \ ,
\qquad\qquad 
|\bs{U}\rangle= 
\begin{bmatrix}
U_{1}(q^1) \\ 
\vdots  \\
U_n(q^n) 
\end{bmatrix} \ ,
\eeq 
where $S_{ij}$ and $U_i$ are arbitrary functions of their separated arguments that satisfy the conditions
\beq
g^{kk}(\boldsymbol{q})=[\bs{S}^{-1}]_{1k} \ , \qquad\qquad V(\bs{q})=\sum_{i=1}^n g^{ii}(\boldsymbol{q})\, U_i(q^i) \ .
\eeq
 Thus, a completely integrable Hamiltonian system of the form \eqref{eq:1.9} is totally separable in a coordinate chart if and only if the following SE 
\beq
\f{1}{2}\bigg(\frac{d W_i}{d q^i} \bigg)^2+U_i(q^i) - \sum_{j=1}^{n} S_{ij}(q^i) H_j= 0 
\eeq
hold. A crucial result in the subsequent analysis is  Benenti's Theorem \cite{Ben80,BeNM89}, stating that
a family of Hamiltonian functions $(H_1,\ldots, H_n)$ are totally separable in a set of
canonical coordinates $(\boldsymbol{q},\boldsymbol{p})$ if and only if  they satisfy the relations
\begin{equation} \label{eq:SI}
\{H_i,H_j\}\lvert_k=\frac{\partial H_i}{\partial q^k}
\frac{\partial H_j}{\partial p_k}-\frac{\partial H_i}{\partial p_k}
\frac{\partial H_j}{\partial q^k}=0 \, , \quad 1\le k \le n
\end{equation}
where no summation over $k$ is understood. In this case, we shall say that the Hamiltonians $(H_1,\ldots, H_n)$ are in \textit{total separable involution} or $\mathcal{T}$-involution (this notion was originally named ``separable involution'' or $\mathcal{S}$-involution in \cite{BeNM89}).

In \cite{AKN1997}, Arnold, Kozlov and Neishtadt (AKN) proposed a generalization of St\"ackel systems by considering the case of arbitrary, not necessarily quadratic, Hamiltonian functions defined in a symplectic manifold. Precisely, given a St\"ackel matrix $\bs{S}$, their construction consists in the family of Hamiltonians
\begin{equation}
H_{j} = \sum_{k=1}^{n} \dfrac{\tilde{S}_{jk}}{\det \bs{S}} \, f_{k}(q^k,p_k), \qquad k=1,\ldots, n
\label{AKN}
\end{equation}
where $\tilde{S}_{jk}$ is the cofactor associated to the element $S_{kj}$, and $f_{k}$ are arbitrary functions. It can be shown that the functions $H_j$, defined by Eqs. \eqref{AKN}, are in involution with each other; consequently, they represent a totally separable Hamiltonian system. A new, recent approach to totally separable systems in the framework of St\"ackel geometry has been developed in \cite{MPhysD2023}.

In many respects, the \textit{bi-Hamiltonian} approach represents one of the most relevant geometric frameworks for the study of Hamiltonian systems. A bi-Hamiltonian manifold is a differentiable manifold endowed with a pencil of Poisson structures \cite{Magri78}. In these settings, the theory of \textit{$\omega N$ manifolds}, introduced in \cite{MMQUAD1984,FPMPAG2003}, is especially relevant in the construction of separation variables for the Hamilton-Jacobi equations associated with a Hamiltonian integrable system. Such manifolds carry on a symplectic structure $\omega$ (assigned as the inverse of a non-degenerate Poisson bivector) and a compatible $(1,1)$-tensor field $\boldsymbol{N}$, also called the recursion or hereditary operator, whose  Nijenhuis torsion \cite{Nij1951} vanishes. This theory has been applied in  \cite{TTPRE2012} to the study of multiseparable systems.

A Haantjes operator is a  $(1,1)$-tensor field whose Haantjes torsion vanishes \cite{Haa1955}. The notion of \textit{Haantjes manifolds} has been introduced by F. Magri in \cite{MGall13} and further studied in \cite{MGall15} and \cite{MGall17}, where the geometry of Haantjes operators has been related with classical mechanics. Also, Haantjes geometry has been studied in connection with the theory of hydrodynamic-type systems \cite{FeKhu,FP2022}. 

\par

Inspired by Magri's construction, in \cite{TT2022AMPA} the family of \textit{symplectic-Haantjes}  (or $\omega \mathscr{H}$) manifolds was introduced. They are symplectic manifolds endowed with a compatible Haantjes algebra  \cite{TT2021JGP}, namely a set of Haantjes $(1,1)$-tensor fields closed under linear $C^{\infty}(M)$-combinations and products. In \cite{KS2017}, a comparison between Magri's Haantjes manifolds and $\omega \mathscr{H}$ manifolds has been made by Kosmann-Schwarzbach.

We recall here the main properties of $\omega\mathscr{H}$ manifolds.
Under suitable hypotheses,  given a Hamiltonian integrable system there exists an $\omega\mathscr{H}$ manifold associated (as stated in the Liouville-Haantjes theorem  \cite{TT2022AMPA}). Many examples of these structures have been exhibited in \cite{TT2016SIGMA}, \cite{TT2022AMPA}, and \cite{RTT2022CNS}. The classical theory of separation of variables can also be formulated in the context of Haantjes geometry. 
The Jacobi-Haantjes theorem, stated in \cite{RTT2022CNS}, establishes under mild hypotheses the equivalence between the notions of total separability of a Hamiltonian system and the existence of a semisimple, Abelian $\omega \mathscr{H}$ manifold of \textit{maximal rank}. Indeed, in the semisimple case, there exists a special class of canonical coordinates, the \textit{Darboux-Haantjes} coordinates, which represent separation variables for the HJ equations and diagonalize simultaneously all the operators of the Haantjes algebra associated with a totally separable model \cite{RTT2022CNS}.

The related notion of \textit{partial separability} refers to the existence of a set of coordinates allowing one to separate the HJ equations into a system of both ODEs and PDEs, the latter ones involving just a subset of independent variables. This relevant subject, in contrast to the theory of total separability, has received so far much less attention. The first study on partial separability dates back to 1896 and is due to G. Di Pirro \cite{DPAMPA1896}, who considered a generalization of the notion of St\"ackel matrix. In his construction, it was assumed that a single row of the $(n-r)\times (n-r)$ generalized St\"ackel matrix, say the first one, depends on a \textit{block of coordinates} $(q^1,\ldots,q^r)$, whereas the remaining rows depend each on one coordinate only, picked up from the set $(q^{r+1},\ldots,q^{n})$. In this way, Di Pirro defined new Hamiltonian systems admitting $n-r$ first integrals, which are homogeneous,   quadratic in the momenta, and orthogonally separated.

Taking into account these results, in 1897, St\"ackel generalized his original construction \cite{StCRAS1893} by proposing a block-St\"ackel invertible matrix, in which more than one row may depend on a subset of coordinates (which are different for each row). The number of quadratic integrals so obtained coincides with the number of subsets of coordinates.
Further results have been contributed by Havas \cite{HavasJMP1976}, who considered for the first time partial separation of the HJ equation in non-orthogonal coordinate systems. Subsequently,  Boyer et al. obtained in \cite{BKMCMP1978} a complete list of all partial separable coordinate systems for the Helmholtz and the HJ equations in non-orthogonal coordinates for the case of four-dimensional Riemannian spaces. Recently, in \cite{CR2019SIGMA} the notion of partial separability has been discussed in the context of Riemannian geometry; in \cite{DKNSIGMA2019}, partially separated Painlev\'e metrics and their conformal deformations have been related to the classical inverse Calder\'on problem.

\subsection{Main results}

The purpose of this article is to contribute to the study of the geometry of separable systems by proving that symplectic-Haantjes manifolds have a relevant role in the theory of partial separability of the HJ equations associated with a Hamiltonian system.

From a conceptual point of view, we shall generalize the Jacobi-Haantjes theorem of \cite{RTT2022CNS} in two different directions. Since for a Hamiltonian system, the existence of semisimple and maximal-rank (Abelian) Haantjes algebras of operators is equivalent to the total separability of the associated HJ equations, in order to characterize partially separable systems we renounce either the hypothesis of semisimplicity or that of maximal rank.

In the first scenario, we deal with a complete family of Hamiltonian functions in involution  $(H_1, \ldots, H_n)$ for which, in a given set of coordinates, a \textit{partial involutivity condition} ($\mathcal{P}$-involution) holds. This means that in these coordinates, the Poisson bracket between two integrals splits into several equations, which generalize Benenti's conditions \eqref{eq:SI} (see Definition \ref{def:BP} below). As we shall prove in Theorem \ref{theo:5}, this, in turn, implies that the Hamiltonian family is partially separable.

Our main result is Theorem \ref{th:Nonsemi}. Given the family $(H_1, \ldots, H_n)$, we shall prove that if one can associate it with an Abelian algebra of \textit{non-semisimple Haantjes operators of maximal rank}, then the Darboux-Haantjes coordinates (allowing us to write all the operators simultaneously in a block-diagonal form) also provide us with a set of partially separable coordinates for the associated HJ equations. 
This result paves the way for the systematic application of non-semisimple Haantjes algebras in classical mechanics.

In the second scenario, we shall deal with \textit{non-maximal-rank}, but semisimple Haantjes algebras. We propose two results. First, we construct a novel class of partially separated systems, which we obtain by further generalizing the AKN construction to the case of block-St\"ackel matrices. We recover as particular instances of our construction all the partially separable systems that have been proposed in the literature so far. Second, we interpret our novel class of Hamiltonian systems in the framework of Haantjes geometry. Precisely, in Theorem \ref{Th4}, we prove the existence of non-maximal-rank, semisimple $\omega \mathscr{H}$ manifolds for these systems, and we provide the explicit expression of their Haantjes algebras.

As a general comment, we wish to clarify that partial and total separability are not exclusive notions. As a matter of fact, a given integrable Hamiltonian system may admit more than one $\omega \mathscr{H}$ structure, which leads either to total or to partial separability depending on whether the structure is semisimple or non-semisimple. 
From a computational point of view, partial separability can also be regarded as a useful intermediate step toward the construction of total-separation variables.

Another relevant point is that the new systems coming from St\"ackel-type matrices represent an (a priori) \textit{incomplete} family of commuting Hamiltonians: we have as many Hamiltonians as the number of sets of variables involved in the partial separation process. We mention that the original constructions of partially separable systems by Di Pirro \cite{DPAMPA1896}, St\"ackel \cite{StAMPA1897} and Painlev\'e \cite{DKNSIGMA2019} lead as well to incomplete sets of Hamiltonian systems. 

However, we prove in Theorem \ref{th:Sinvol} that for our St\"ackel-type Hamiltonian systems, still the Benenti conditions of $\mathcal{T}$-involution \eqref{eq:SI} hold in the set of coordinates in which the family is defined.

To summarize, Haantjes geometry provides us with two general classes of partially separated Hamiltonian systems. The first class possesses non-semisimple, maximal rank Haantjes algebras associated. These Hamiltonian systems are completely integrable and are in $\mathcal{P}$-involution. The associated HJ equations are partially separable in Darboux-Haantjes coordinates. 

The second class involves ``incomplete'' families of Hamiltonian systems: the number of integrals explicitly known is less than $n$, half of the dimension of the phase space. These systems possess semisimple, non-maximal-rank Haantjes algebras, are in standard $\mathcal{T}$-involution, and, once again, the associated HJ equations are partially separable in Darboux-Haantjes coordinates. 
Our main conclusion is that symplectic-Haantjes manifolds represent a natural geometric setting for studying partially separated Hamiltonian systems.

The structure of the article is the following. In Section \ref{sec:2}, we review standard results concerning Haantjes geometry and the theory of symplectic-Haantjes manifolds. In Section \ref{sec:3}, we define the notion of partial separability and a generalization of Benenti conditions, namely the concept of $\mathcal{P}$-involution.
 In Section \ref{sec:4}, we address the problem of partial separability for completely integrable models; we shall show that the Darboux-Haantjes coordinates which block-diagonalize the non-semisimple Haantjes algebra associated with a given integrable model also guarantee partial separability of the associated HJ equations. In Section \ref{sec:5}, we define a new class of generalized St\"ackel systems and show that they admit semisimple, non-maximal-rank symplectic-Haantjes manifolds. In Section \ref{sec:6}, two examples illustrating the application of the theory proposed are presented. Some open problems and future perspectives are outlined in the final Section \ref{sec:7}.

\section{Preliminaries on the Haantjes geometry} \label{sec:2}

\subsection{Basic notions}
We shall start by reviewing the basic definitions concerning the geometry of the standard Nijenhuis and Haantjes torsions, following the original, pioneering papers \cite{Haa1955,Nij1951}  and the related ones
\cite{Nij1955,FN1956}.

Let $M$ be a differentiable manifold of dimension $n$, $\mathfrak{X}(M)$ be the Lie algebra of vector fields on $M$, $\mathfrak{X}^{*}(M)$ the module of one-forms on $M$ and  $\boldsymbol{A}:\mathfrak{X}(M)\rightarrow \mathfrak{X}(M)$ be a  $(1,1)$-tensor field (i.e. a linear operator field). Henceforth, all tensor and vector fields will be considered to be smooth.
\begin{definition} \label{def:N}
The
 \textit{Nijenhuis torsion} of $\boldsymbol{A}$ is the vector-valued $2$-form defined by
\begin{equation} \label{eq:Ntorsion}
\tau_ {\boldsymbol{A}} (X,Y):=\boldsymbol{A}^2[X,Y] +[\boldsymbol{A}X,\boldsymbol{A}Y]-\boldsymbol{A}\Big([X,\boldsymbol{A}Y]+[\boldsymbol{A}X,Y]\Big),
\end{equation}
where $X,Y \in \mathfrak{X}(M)$ and $[ \, \cdot \, , \cdot \, ]$ denotes the commutator of two vector fields.
\end{definition}
\begin{definition} \label{def:H}
 \noi The \textit{Haantjes torsion} of $\boldsymbol{A}$ is the vector-valued $2$-form defined by
\begin{equation} \label{eq:Haan}
\mathcal{H}_{\boldsymbol{A}}(X,Y):=\boldsymbol{A}^2 \tau_{\boldsymbol{A}}(X,Y)+\tau_{\boldsymbol{A}}(\boldsymbol{A}X,\boldsymbol{A}Y)-\boldsymbol{A}\Big(\tau_{\boldsymbol{A}}(X,\boldsymbol{A}Y)+\tau_{\boldsymbol{A}}(\boldsymbol{A}X,Y)\Big).
\end{equation}
\end{definition}

\begin{definition}
A Haantjes (Nijenhuis) operator is a (1,1)-tensor field whose Haantjes (Nijenhuis) torsion identically vanishes.
\end{definition}

One of the simplest and most relevant cases of Haantjes operators consists of the family of tensor fields $\boldsymbol{A}$ which take a diagonal form in suitable coordinates $\boldsymbol{x}=(x^1,\ldots,x^n)$:
\begin{equation}
\boldsymbol{A}(\boldsymbol{x})=\sum _{i=1}^n \lambda_{i }(\boldsymbol{x}) \frac{\partial}{\partial x^i}\otimes \rd x^i \ . \label{eq:Ldiagonal}
 \end{equation}
Here $\lambda_{i }(\boldsymbol{x}):=\lambda^{i}_{i}(\boldsymbol{x})$ are the eigenvalue fields of $\boldsymbol{A}$ and  $\left(\frac{\partial}{\partial x^1},\ldots, \frac{\partial}{\partial x^n}\right) $ are the fields forming the \textit{natural frame} associated with the local chart $(x^{1},\ldots,x^{n})$. As is well known, the Haantjes torsion of a diagonal operator of the form \eqref{eq:Ldiagonal} vanishes. Furthermore, if $\lambda_{i }(\boldsymbol{x}) \equiv \lambda_{i }(x^{i})$, then $\boldsymbol{A}$ is also a Nijenhuis operator.

\vspace{3mm}

In the following, $\boldsymbol{I}: \mathfrak{X}(M)\rightarrow \mathfrak{X}(M)$ will denote the identity operator.

\begin{proposition}  \label{pr:HpolL} \cite{BogIzv2004}.
Let  $\boldsymbol{A}$ be a (1,1)-tensor field. Then, for any polynomial in $\boldsymbol{A}$ with coefficients $c_{k} \in C^{\infty} (M)$, the associated Haantjes torsion vanishes, i.e.
\begin{equation} \label{eq:HpolL}
\mathcal{H}_{\boldsymbol{A}}(X,Y) = 0 \Longrightarrow \mathcal{H}_{\sum_{k} c_{k} (\boldsymbol{x}) \boldsymbol{A}^{k}}(X,Y)=0.
\end{equation}
\end{proposition}

As discussed in \cite{TT2021JGP}, by virtue of Proposition \ref{pr:HpolL}, a Haantjes operator generates a cyclic algebra of Haantjes operators over the ring of smooth functions on $M$. 

Besides,  a novel class of generalized Nijenhuis operators of level $l$ has been introduced in \cite{KS2017,TT2021JGP} and further discussed in \cite{TT2022CMP,TT2022CM}. When $l=1$, they coincide with the standard Nijenhuis operators; for $l=2$, they are Haantjes operators. For $l>2$, we have genuinely new operators. In \cite{RTT2023JNS}, a generalization of Proposition \ref{pr:HpolL} has been proved for the generalized Nijenhuis operators of level $l>2$.

\vspace{2mm}

Haantjes operators play a key role in classical mechanics and Riemannian geometry. This is illustrated for instance in the works \cite{RTT2022CNS,RTT2023JNS}, and \cite{TT2021JGP}--\cite{TT2016SIGMA}.

\subsection{Haantjes algebras}

Another crucial notion is that of Haantjes algebras, which have been introduced in \cite{TT2021JGP} and further studied and generalized in \cite{RTT2023JNS}.

\begin{definition}\label{def:HM}
A Haantjes algebra is a pair    $(M, \mathscr{H})$ satisfying the following properties:
\begin{itemize}
\item
$M$ is a differentiable manifold of dimension $\mathrm{n}$;
\item
$ \mathscr{H}$ is a set of Haantjes  operators $\boldsymbol{K}:\mathfrak{X}(M)\to \mathfrak{X}(M)$. Also, they  generate:
\begin{itemize}
\item
a free module over the ring of smooth functions on $M$:
\begin{equation}
\label{eq:Hmod}
\mathcal{H}_{\left( f\boldsymbol{K}_{1} +
                             g\boldsymbol{K}_2 \right)}(X,Y)= \mathbf{0}
 \, , \qquad\forall\, X, Y \in \mathfrak{X}(M) \, , \quad \forall\, f,g \in C^\infty(M) \,  , \quad \forall ~\boldsymbol{K}_1,\boldsymbol{K}_2 \in  \mathscr{H} ;
\end{equation}
  \item
a ring  w.r.t. the composition operation
\begin{equation}
 \label{eq:Hring}
\mathcal{H}_{(\boldsymbol{K}_1 \, \boldsymbol{K}_2)}(X,Y)=\mathbf{0} \, , \qquad
\forall\, \boldsymbol{K}_1,\boldsymbol{K}_2\in  \mathscr{H} , \quad\forall\, X, Y \in \mathfrak{X}(M)\, .
\end{equation}
\end{itemize}
\end{itemize}
If
\begin{equation}
\boldsymbol{K}_1\,\boldsymbol{K}_2=\boldsymbol{K}_2\,\boldsymbol{K}_1 \, , \quad\qquad\ \forall~\boldsymbol{K}_1,\boldsymbol{K}_2 \in  \mathscr{H} ,
\end{equation}
the  algebra $(M, \mathscr{H})$ is said to be an Abelian  Haantjes algebra. 
\end{definition}

In essence, we can think of $\mathscr{H}$ as an associative algebra of Haantjes operators.

\begin{remark}
A fundamental fact concerning Abelian Haantjes algebras is that there exist sets of coordinates, called \textit{Haantjes coordinates}, which allow us to write simultaneously all $\boldsymbol{K}\in \mathscr{H}$ in a block-diagonal form. In particular, if $\mathscr{H}$ is an algebra of semisimple operators, then in Haantjes coordinates all $\boldsymbol{K}\in \mathscr{H}$ acquire a purely diagonal form \cite{TT2021JGP}. 

\end{remark}

\vspace{2mm}

\subsection{The symplectic-Haantjes manifolds}

Symplectic-Haantjes (or $\omega \mathscr{H}$) manifolds are symplectic manifolds endowed with a compatible algebra of Haantjes operators. Apart from their mathematical properties,   these manifolds are relevant because they provide a simple but sufficiently general setting in which the theory of Hamiltonian integrable systems can be naturally formulated.

\begin{definition}\label{def:oHman}
A symplectic--Haantjes (or $\omega \mathscr{H}$) manifold  of class $m$ is a triple $( M,\omega,\mathscr{H})$ which satisfies the following properties:
\begin{itemize}
\item[(i)]
$(M,\omega)$  is a   symplectic  manifold of dimension $ 2 \, n$;
\item[(ii)]
$\mathscr{H}$ is an Abelian Haantjes algebra of rank $m$;
\item[(iii)]
$(\omega,\mathscr{H})$ are algebraically compatible, that is
\beq
\omega(X,\boldsymbol{K} Y)=\omega(\boldsymbol{K} X,Y) \, , \qquad \forall \boldsymbol{K} \in \mathscr{H} ,
\eeq
or equivalently
\begin{equation}\label{eq:compOmH}
\boldsymbol{\Omega}\, \boldsymbol{K} =\boldsymbol{K}^T\boldsymbol{\Omega} \, ,\qquad \ \forall \boldsymbol{K} \in \mathscr{H} .
\end{equation}
\end{itemize}
\noi
 Hereafter $\boldsymbol{\Omega}:=\omega ^\flat:\mathfrak{X}(M) \rightarrow \mathfrak{X}^{*}(M)$ denotes the  fiber bundles isomorphism defined by
\beq
\omega(X,Y)=\langle \boldsymbol{\Omega} X,Y \rangle \, , \qquad\qquad\forall X, Y \in \mathfrak{X}(M).
\eeq
The map $\boldsymbol{\Omega}^{-1}:\mathfrak{X}^{*}(M) \rightarrow \mathfrak{X}(M)$  is the Poisson bivector induced by  the symplectic two-form $\omega$.
\par
\end{definition}

The commutativity condition in $(ii)$ of the latter Definition is motivated by the following
\begin{proposition}
The set of (non-necessarily Haantjes) operators compatible with the symplectic form $\omega$ is a  $C^\infty(M)$ module.  Moreover,  this module is an algebra if and only if it is Abelian with respect to the product of operators.
\begin{proof}
The first statement follows immediately from the linearity of the compatibility condition \eqref{eq:compOmH}. Concerning the second one, we note that if $\bs{K}_1$, $\bs{K}_2$ are operators such that
\beq
\boldsymbol{\Omega}\boldsymbol{K}_1=\boldsymbol{K}_1^T \boldsymbol{\Omega},\qquad
\boldsymbol{\Omega}\boldsymbol{K}_2=\boldsymbol{K}_2^T \boldsymbol{\Omega},
\eeq
then 
\beq
\boldsymbol{\Omega}\boldsymbol{K}_1\boldsymbol{K}_2=\boldsymbol{K}_1^T \boldsymbol{\Omega} \boldsymbol{K}_2=
\boldsymbol{K}_1^T\boldsymbol{K}_2^T\boldsymbol{\Omega}=(\boldsymbol{K}_2\boldsymbol{K}_1)^T\boldsymbol{\Omega} \ .
\eeq
Thus, $\bs{K}_1\bs{K}_2$ is compatible with the symplectic form $\omega$ if and only if $\bs{K}_1$ and $\bs{K}_2$ commute.
\end{proof}
\end{proposition}
We also remind the following
\begin{proposition}\label{lm:aupari} \cite{TT2022AMPA}
Given a $2n$-dimensional $\omega \mathscr{H}$ manifold $M$, every generalized eigen--distribution  $ker(\boldsymbol{K}-\lambda_{i } \boldsymbol{I})^{{\rho_i}}$, $\rho_i \in \mathbb{N}$, is  of even rank.
Therefore, the geometric and algebraic multiplicities of each eigenvalue $\lambda_{i }(\boldsymbol{x})$ are even.
\end{proposition}

Let us consider the spectral decomposition of the tangent spaces of $M$ given by
\begin{equation}\label{def:Di}
T_{\boldsymbol{x}}M=\bigoplus_{i=1}^s \mathcal{D}_i(\boldsymbol{x}), \qquad \qquad
\mathcal{D}_i(\boldsymbol{x}) :=
ker\big(\boldsymbol{K} - \lambda_i\boldsymbol{I}\big)^{\rho_i}(\boldsymbol{x}),
\end{equation}
where $\rho_i$ is the Riesz index of the eigenvalue $\lambda_i$. Due to the latter Proposition,  as $rank$ $\mathcal{D}_i\geq 2, \  i=1,\ldots s$ we conclude that the number of distinct eigenvalues of any
Haantjes operator $\boldsymbol{K}$ of an $\omega\mathscr{H}$ manifold is not greater than $n$.

\subsection{Haantjes chains}
The theory of Magri chains is a fundamental piece of the geometric approach to soliton hierarchies. Magri chains have been introduced in order to construct integrals of motion in
involution for infinite-dimensional Hamiltonian systems \cite{Magri78}.

\noi Hereafter, we remind the notion of Haantjes chains. It has been introduced in \cite{TT2022AMPA} and will be extensively used in the forthcoming analysis.

\begin{definition}
 Let $( M,\mathscr{H})$ be a Haantjes algebra of rank $m$. We shall say that a smooth function $H$ generates a Haantjes chain of closed 1-forms of length $m$ if  there exists
a distinguished basis  $\{\tilde{\boldsymbol{K}}_1,\ldots, \tilde{\boldsymbol{K}}_m\}$ of $\mathscr{H}$
 such that
\begin{equation} \label{eq:MHchain}
\rd (\tilde{\boldsymbol{K}}^T_\alpha \,\rd H )=\boldsymbol{0} \ , \quad\qquad \alpha=1,\ldots ,m \ ,
\end{equation}
where $\tilde{\boldsymbol{K}}^{T}_{\alpha}: \mathfrak{X}^{*}(M) \to \mathfrak{X}^{*}(M)$ is the transposed operator of $\tilde{\boldsymbol{K}}_{\alpha}$. The (locally) exact 1-forms 
\beq
\rd H_\alpha=\tilde{\boldsymbol{K}}^T_\alpha \,\rd H,
\eeq
which are supposed to be linearly independent, are said to be the elements of the Haantjes chain of length $m$ generated by $H$. The functions $H_\alpha$ are the potential functions of the chain.
\end{definition}

In order to enquire about the existence of Haantjes chains for an assigned Haantjes algebra, we have to consider the co-distribution, of rank  $r\leq m$, generated by a given function $H$ via an arbitrary basis
$ \{ \boldsymbol{K}_1,  \boldsymbol{K}_2,\ldots,\boldsymbol{K}_{m}\} $ of $\mathscr{H}$, i.e.
\begin{equation} \label{eq:codKH}
\mathcal{D}_H^\circ:=Span\{ \boldsymbol{K}_1^T dH,  \boldsymbol{K}_2^T \rd H,\ldots,\boldsymbol{K}_{m}^T\,\rd H\} \ ,
\end{equation}
and the distribution $\mathcal{D}_H$ of the vector fields annihilated by them, which has rank  $n-r$.
Note that such distributions do not depend on the particular choice of the basis of $\mathscr{H}$.
The following theorems offer a geometric characterization of the existence of a Haantjes chain generated by a smooth function $H$ in terms of the Frobenius integrability of its associated co-distribution.
\begin{theorem} [Theorem 16, \cite{TT2022AMPA}] \label{th:LHint}
 Let  $(M,  \mathscr{H})$ be a  Haantjes algebra of rank $m$, and  $H$  a smooth function on $M$. Let  $\mathcal{D}_H^\circ$ be the codistribution \eqref{eq:codKH},
assumed of rank $m$ (independent on $\boldsymbol{x}$),
and let $\mathcal{D}_H$ be the distribution of the vector fields annihilated by the $1$-forms of $\mathcal{D}_H^\circ$.
Then, the function $H$  generates a Haantjes chain $\eqref{eq:MHchain}$
if and only if $\mathcal{D}^\circ_H$ (or equivalently $\mathcal{D}_H$) is Frobenius-integrable.
\par
\end{theorem}

  The following result clarifies the compatibility between a given $\omega  \mathscr{H}$ manifold and a set of functions in involution.
\begin{theorem} [Theorem 35, \cite{TT2022AMPA}] \label{pr:LHDe}
Let   $(M, \omega,  \mathscr{H})$ be a
  $2n$-dimensional $\omega \mathscr{H}$ Abelian manifold of class $m$. Let $(H_1,H_2,\ldots, H_m)$
	be $m$ independent functions  in involution and  $\mathcal{D}^\circ=Span\{\rd H_1, \rd H_2, \ldots,\rd H_m \}$.  The differentials $(\rd H_1,\rd H_2,\ldots, \rd H_m)$  form  a Haantjes chain generated by a smooth function $H$ in involution with  them  if and only if $H$  satisfies the condition
  \begin{equation} \label{eq:LtH}
\mathcal{D}_H^\circ = \mathcal{D}^\circ\ .
  \end{equation}
\end{theorem}

We also propose now a new theorem, which characterizes the Hamiltonian functions generating Haantjes chains.
\begin{theorem} \label{th:DLinv}
Let $(H_1,H_2,\ldots, H_m)$
	be $m$ independent functions  in involution and  $\mathcal{D}^\circ=Span\{\rd H_1, \rd H_2, \ldots,\rd H_m \}$. The set $(H_1,H_2,\ldots, H_m)$ are the potential functions of a Haantjes chain formed by
$(\rd H_1,\rd H_2,\ldots, \rd H_m)$ if and only if,  for all $\boldsymbol{K} \in \mathscr{H}$, the following invariance condition holds
\begin{equation}\label{eq:Linv}
\boldsymbol{K}^T (\mathcal{D}^\circ ) \subseteq \mathcal{D}^\circ .
\end{equation}
\end{theorem}
\begin{proof}
Let $H$ be functionally dependent on $(H_1,\ldots,H_m)$. Then $\rd H \in \mathcal{D}^\circ$, and $H$ is in involution with $(H_1,\ldots,H_m)$. Moreover, for each $\boldsymbol{K} \in  \mathscr{H}$ we have
\beq
\boldsymbol{K}^T \rd H=\boldsymbol{K}^T \bigg(\sum_{i=1}^m f_i \, \rd H_i \bigg)=\sum_{i=1}^m f_i \,\boldsymbol{K}^T \rd H_i.
\eeq
If the invariance condition \eqref{eq:Linv} is fulfilled, then $\mathcal{D}_H^\circ \subseteq \mathcal{D}^\circ$; in particular, having the same rank by assumption, the two distributions coincide. Thus, we deduce that condition \eqref{eq:LtH} is satisfied. 
Conversely,  if condition \eqref{eq:LtH} is satisfied, then $\boldsymbol{K}^T \rd H$ is a linear combination of $(\rd H_1,\ldots, \rd H_m)$. Thus, since $\rd H \in \mathcal{D}^\circ$, we conclude that the invariance condition \eqref{eq:Linv} is fulfilled.
\end{proof}
From Theorems \ref{pr:LHDe} and \ref{th:DLinv}, we deduce the following 
\begin{corollary}
A function  $H$ generates a Haantjes chain only if it satisfies the invariance condition
\begin{equation}
\boldsymbol{K}^T (\mathcal{D}_H^\circ ) \subseteq \mathcal{D}_H^\circ \qquad\qquad \qquad\forall \boldsymbol{K} \in \mathscr{H}\  .
\end{equation}
\end{corollary}
\subsection{Haantjes coordinates}
As we shall see, a crucial result for the theory of separability of Hamiltonian systems is the theorem, proved in \cite{TT2021JGP},  ensuring that an algebra of Haantjes operators can be represented in suitable coordinates in a block-diagonal form.

\begin{proposition} [Proposition 33, \cite{TT2021JGP}]  \label{th:HJ}
Let  $\mathcal{K}=\{\boldsymbol{K}_1,\ldots,\boldsymbol{K}_m\}$, $\boldsymbol{K}_\alpha: \mathfrak{X}(M)\to \mathfrak{X}(M)$, $\alpha=1,\ldots,m$ be  a family  of commuting   operator fields; we assume that  one of them, say $\boldsymbol{K}_1$, is a Haantjes operator.  Then, there exist local charts adapted to the spectral decomposition of $\bs{K}_1$ in which all of the operators $\boldsymbol{K}_\alpha$ can be written simultaneously in a block-diagonal form. In addition, if we assume that

(i) all the operators of the family are Haantjes operators

(ii) all possible nontrivial intersections of their generalized eigen-distributions
\begin{equation}
\label{eq:Va}
\mathcal{V}_a(\boldsymbol{x}):=  \bigoplus_{i_1,\ldots,i_m}^{s_1,\ldots,s_m}\mathcal{D}_{i_1}^{(1)}(\boldsymbol{x}) \bigcap  \ldots      \bigcap \mathcal{D}_{i_m}^{(m)}(\boldsymbol{x}) \ ,\qquad a=1,\ldots, v\leq n
\end{equation}
are pairwise integrable (i.e. any direct sum of $\mathcal{V}_a$ is integrable), then  there exist  sets of local coordinates, adapted to the decomposition
\begin{equation}
 \label{eq:TVa}
T_{\boldsymbol{x}}M= \bigoplus_{a=1}^{v}\mathcal{V}_a( \boldsymbol{x}) \qquad  \boldsymbol{x} \in M ,
\end{equation}
in which all operators $\boldsymbol{K}_\alpha$ admit simultaneously a block-diagonal form (with possibly finer blocks).
\end{proposition}

A simple but relevant consequence of the latter Proposition is the following new result.
\begin{corollary}
The property of simultaneous block-diagonalization stated in Proposition \ref{th:HJ} extends to the algebra generated by the family of operators $\mathcal{K}$ by considering $C^{\infty}(M)$-linear combinations and products of operators. 
\end{corollary}
\begin{remark}
Each set of coordinates allowing us to write down the family $\mathcal{K}$ in a block-diagonal form will be said to be a set of \textit{Haantjes coordinates}. In particular, if we deal with a semisimple Haantjes algebra, then each set of Haantjes coordinates allows us to write simultaneously all the operators of the algebra in a diagonal form.
\end{remark}
This theorem has been extended in \cite{RTT2023JNS} to the case of \textit{generalized Haantjes algebras of level} $l$, which are algebras whose elements are represented by generalized Nijenhuis operators of level $l$.
\begin{definition}
Given an $\omega \mathscr{H}$ manifold, the local coordinates where all Haantjes operators take simultaneously a block-diagonal form and the symplectic form takes the Darboux form are called Darboux--Haantjes (DH) coordinates.
\end{definition}
The existence of DH coordinates for a Haantjes manifold is guaranteed by Proposition \ref{th:HJ} and the following reasoning. Precisely, the fact that all the operators considered are Haantjes operators ensures that there exist coordinates allowing for a simultaneous block-diagonalization. Notice that all eigendistributions $\mathcal{V}_a(\bs{x})$ are integrable, as they are intersections of integrable distributions. Moreover, they are of even rank, as proved in \cite{TT2022AMPA}. If they are also pairwise integrable, then the block-diagonalization can be realized with finer blocks. 
By definition, the operators are also algebraically compatible with the symplectic form. As each of these blocks is invariant, we can restrict the symplectic form on them.
Then, over the leaves of each distribution $\mathcal{V}_a$, one can find Darboux coordinates for the corresponding restriction of the symplectic form, which is still symplectic. This is due to the fact that the distributions are integrable and of even rank, so their integral leaves are symplectic submanifolds and are symplectically orthogonal to each other \cite{TT2022AMPA}. Therefore, one can collect  such coordinates to obtain a local chart in $M$ of the form
\begin{equation} \label{eq:gDHch}
(\boldsymbol{q}^{1},\boldsymbol{p}_{1},\ldots,\boldsymbol{q}^{v},\boldsymbol{p}_{v}) \ ,
\end{equation}
where the subsets
$$
(\boldsymbol{q}^{1},\boldsymbol{p}_{1})=(q^{1,1},\ldots,q^{1,\sigma_{1}}, p_{1,1},\ldots,p_{1,\sigma_{1}}); \ldots;
 (\boldsymbol{q}^{v},\boldsymbol{p}_{v})=(q^{v,1},\dots, q^{v,\sigma_{v}},p_{v,1},\dots, p_{v,\sigma_{v}}),
$$
are adapted to the decomposition \eqref{eq:TVa}. In this chart, the symplectic form $\omega$ takes  the Darboux form \begin{equation} \label{eq:gDo}
\omega=
\sum_{a=1}^{v}\sum_{i_a=1}^{\sigma_a}  \rd p_{a,i_a} \wedge \rd   q^{a,i_a} \ .
\end{equation}
Here 
$
\sigma_{a}:= \frac{1}{2} \, \text{rank} \, (\mathcal{V}_a).
$
 Also, as a consequence of Proposition \ref{th:HJ}, each Haantjes operator $\boldsymbol{K} \in \mathscr{H}$  possesses a block-diagonal form. 
Then, in the chart \eqref{eq:gDHch} the operators $\boldsymbol{K}_{\alpha} \in \mathscr{H}$ can be written as
\begin{equation} \label{eq:KNS}
\begin{split}
\boldsymbol{K}_{\alpha} =& \sum_{a = 1}^{v}
\sum_{j_{a}, k_{a} = 1}^{\sigma_{a}} 
\bigg(
 \left(
 \bs{A}_{a}^{(\alpha)}(\boldsymbol{q}, \boldsymbol{p})
 \right)^{ a, j_{a}}_{ a, k_{a}} 
   \left( 
  \dfrac{\partial}{\partial q^{a, j_{a}}} \otimes d q^{a, k_{a}} \right) +
   \left(
 \bs{D}_{a}^{(\alpha)}(\boldsymbol{q}, \boldsymbol{p})
 \right)^{ a, j_{a}}_{ a, k_{a}} 
 \left(
 \dfrac{\partial}{\partial p_{a, j_{a}}} \otimes d p_{a, k_{a}} 
\right) 
\\
&+
 \left(
 \bs{B}_{a}^{(\alpha)}(\boldsymbol{q}, \boldsymbol{p})
 \right)^{ a, j_{a}}_{ a, k_{a}} 
  \left(
   \dfrac{\partial}{\partial q^{a, j_{a}}} \otimes d p_{a, k_{a}} 
    \right )
 +
 \left(
 \bs{C}_{a}^{(\alpha)}(\boldsymbol{q}, \boldsymbol{p})
 \right)^{ a, j_{a}}_{ a, k_{a}} 
  \left(
   \dfrac{\partial}{\partial p_{a, j_{a}}} \otimes d q^{a, k_{a}}
   \right )
\bigg)
\end{split}
\end{equation}
where the blocks $ \boldsymbol{A}_{a}^{(\alpha)}, \boldsymbol{B}_{a}^{(\alpha)}, \boldsymbol{C}_{a}^{(\alpha)}, \boldsymbol{D}_{a}^{(\alpha)}$ are $\sigma_a \times \sigma_a$ matrix-valued functions suitably chosen in order for $\boldsymbol{K}_{\alpha}$ to have a vanishing Haantjes torsion. Due to the compatibility condition \eqref{eq:compOmH},  the relations
\begin{equation} \label{eq:cOmNS}
 ( \boldsymbol{D}_{a}^{(\alpha)})^T= \boldsymbol{A}_{a}^{(\alpha)}\ , \quad 
 \boldsymbol{B}_{a}^{(\alpha)}+ (\boldsymbol{B}_{a}^{(\alpha)})^T=0\ ,\quad 
  \boldsymbol{C}_{a}^{(\alpha)}+ (\boldsymbol{C}_{a}^{(\alpha)})^T=0 \quad\quad \alpha=1,\ldots,m \ ,
\end{equation}
must also be satisfied. Moreover, as the Haantjes algebra $\mathscr{H}$ is Abelian,  the abovementioned blocks have to satisfy  the further conditions
\begin{align}
\label{eq:AAcomm} 
&[ \boldsymbol{A}_{a}^{(\alpha)}, \boldsymbol{A}_{a}^{(\beta)}]+ 
\boldsymbol{B}_{a}^{(\alpha)} \boldsymbol{C}_{a}^{(\beta)}- \boldsymbol{B}_{a}^{(\beta)} \boldsymbol{C}_{a}^{(\alpha)}=0 \ , \\
\label{eq:ABskew}
&\boldsymbol{A}_{a}^{(\alpha)}\boldsymbol{B}_{a}^{(\beta)}-
\boldsymbol{A}_{a}^{(\beta)}\boldsymbol{B}_{a}^{(\alpha)}+
\boldsymbol{B}_{a}^{(\alpha)}(\boldsymbol{A}_{a}^{(\beta)})^T-
\boldsymbol{B}_{a}^{(\beta)}(\boldsymbol{A}_{a}^{(\alpha)})^T=0 \ ,\\
\label{eq:ACskew}
&\boldsymbol{C}_{a}^{(\alpha)}\boldsymbol{A}_{a}^{(\beta)}-
\boldsymbol{C}_{a}^{(\beta)}\boldsymbol{A}_{a}^{(\alpha)}
+(\boldsymbol{A}_{a}^{(\alpha)})^T\boldsymbol{C}_{a}^{(\beta)} -
(\boldsymbol{A}_{a}^{(\beta)})^T\boldsymbol{C}_{a}^{(\alpha)}=0 \ .
\end{align}

Combining the conditions \eqref{eq:cOmNS} with Eqs. \eqref{eq:ABskew} and  \eqref{eq:ACskew}, we deduce a useful result.
\begin{lemma} \label{lm:ABACskew}
The operators 
\begin{equation} \label{eq:ABACskew}
\boldsymbol{A}_{a}^{(\alpha)}\boldsymbol{B}_{a}^{(\beta)}-\boldsymbol{A}_{a}^{(\beta)}\boldsymbol{B}_{a}^{(\alpha)}\ , \quad\qquad
\boldsymbol{C}_{a}^{(\alpha)}\boldsymbol{A}_{a}^{(\beta)}-\boldsymbol{C}_{a}^{(\beta)}\boldsymbol{A}_{a}^{(\alpha)}
\end{equation}
are skew-symmetric.
\end{lemma}
In particular, if the Haantjes algebra $\mathscr{H}$ is \textit{semisimple}, the blocks $ \boldsymbol{B}^{(\alpha)}_{a}$ and $ \boldsymbol{C}^{(\alpha)}_{a}$ vanish, whereas the blocks $ \boldsymbol{A}^{(\alpha)}_{a}$ take the form $\boldsymbol{A}^{(\alpha)}_{a}= \lambda^{(\alpha)}_{a} \bs{I}_{\sigma_{a}}$, where
$\lambda^{(\alpha)}_{a}$ are the eigenvalues of $\boldsymbol{K}_\alpha$.
\section{Partial separability} \label{sec:3}

\subsection{Main definition}
We shall propose a definition of partial separability inspired by the original theory of Jacobi. In essence, we require that, given a Hamiltonian system, the corresponding Hamilton-Jacobi equations decouple into a system of PDEs, each of them depending on a different subset of independent variables. More precisely, we have the following
\begin{definition}[Generalized Separation Equations] \label{def:PS}
We shall say that a $2n$-dimensional Hamiltonian system with Hamiltonian function $H$ is partially separable if there exist:

(1) $m\leq n$ independent smooth functions $(H_1,\ldots, H_m)$ in involution, on which $H$ is functionally dependent;

(2) $1<v\leq n$ subsets of Darboux coordinates $(\boldsymbol{q}^{1},\boldsymbol{p}_{1},\ldots,\boldsymbol{q}^{v},\boldsymbol{p}_{v})$ defined in the phase space;

(3) $m$ separation equations of the form

\begin{equation} \label{eq:gSk}
\phi_{a,j_a}(\boldsymbol{q}^{a},\boldsymbol{p}_{a}; H_1,\ldots,H_m)=0 \ , \quad   a=1,\ldots,v,\qquad j_a=1,\ldots r_a, \quad  {\rm  det}\left[\frac{\partial \boldsymbol{\phi}}{\partial \boldsymbol{H} }\right] \neq 0
 \end{equation} 
where $r_a\in \mathbb{N} \backslash \{0\}$ represents the number of equations associated with
$$
(\boldsymbol{q}^{a},\boldsymbol{p}_{a})=(q^{a,1},\ldots,q^{a,\sigma_{a}}, p_{a,1},\ldots,p_{a,\sigma_{a}}), 
$$
$\boldsymbol{\phi}=(\phi_{1,1}\ldots, \phi_{1,r_1},\ldots,\phi_{v,1},\ldots, \phi_{v,r_v})^{T}$, $\boldsymbol{H}=(H_1,\ldots, H_m)$
and $\sigma_a$ represents the number of coordinates (or momenta) involved in the pair $(\boldsymbol{q}^{a},\boldsymbol{p}_{a})$.
\end{definition}
These equations are a generalization of the Jacobi-Sklyanin SE \eqref{eq:Sk} and will be referred to as generalized separation equations (GSE). In fact, if for each $a\in \{1,\ldots,v\}$, $\sigma_a=1$ (which implies $v=n$)  and $r_a=1$ (then $v=m=n$), one recovers the classical  Jacobi's theory of separability. A Hamiltonian system such that $\sigma_a=r_a=1$ will be said to be \textit{totally separable}.

In this article, we shall study two cases in which
$\sigma_a>1$ for some $a\in \{1,\ldots,v\}$ (thus, $v<n$):
  
\vspace{3mm}

\begin{itemize}
\item Case \textbf{I}:
   $r_a=\sigma_a$, implying $m=n$;
   
\vspace{3mm}

\item 
Case \textbf{II}:
  $r_a=1$ $\forall a\in \{1,\ldots,v\}$, implying $m=v$.
\end{itemize}
The Case \tbf{I} corresponds to the maximum possible number of (independent) separation equations for each subset of Darboux coordinates, whereas the Case \tbf{II} corresponds to the minimum one.
Other possible cases, corresponding to different values of the indices $(\sigma_a, r_a)$ are left for future work.

\begin{remark}

Definition \ref{def:PS} can be regarded as a generalization of the approach to partial separability developed in \cite{DPAMPA1896}, \cite{StAMPA1897}, \cite{DKNSIGMA2019}, \cite{HavasJMP1976}, \cite{CR2019SIGMA},  etc. By solving, when possible, the PDEs \eqref{eq:gSk}, one obtains a complete integral shared by the HJ equations of $(H_1,\ldots, H_n)$ which is partially additively separated, i.e., of the form
\begin{equation}\label{eq:CI}
W=\sum_{a=1}^v W_a (\boldsymbol{q}^a; h_1,\ldots,h_n),\qquad \det \bigg( \frac{\partial^2 W}{\partial q^i\partial h_j} \bigg)\neq 0.
\end{equation}

For instance, this is guaranteed  in Case \textbf{I}, under the hypotheses of the following result.

\end{remark}

\begin{theorem}\ \label{th:IPS}
Let $(H_1,\ldots,H_n)$ be a family of $n$ independent functions satisfying the GSE \eqref{eq:gSk} with $r_a=\sigma_a$, and $\mathcal{L}_{\bs{h}}$ be the foliation defined by the level sets $\mathcal{L}_{\bs{h}}:=\{H_1=h_1,\ldots,H_n=h_n\}$. Assuming that
\begin{equation}\label{eq:Jacp}
\det\left[\frac{\partial \boldsymbol{\phi} }{\partial{\boldsymbol{p} }}\right] \neq 0 \ ,
\end{equation} 
the Hamiltonian system $(H_1,\ldots,H_n)$ admits a complete integral of the form \eqref{eq:CI} if and only if the $n \times n$ matrix
\begin{equation} \label{eq:SymPhi}
\left[\frac{\partial \boldsymbol{\phi} } {\partial{\boldsymbol{p} }}\right]^{-1}_{\big|_{{\mathcal{L}}_{\bs{h}}}}
\left[\frac{\partial \boldsymbol{\phi} }{\partial{\boldsymbol{q} }}\right]_{\big|_{{\mathcal{L}}_{\bs{h}}}}=\left[\frac{\partial{\boldsymbol{p}}}{\partial{\boldsymbol{q}}}\right]_{\big|_{{\mathcal{L}}_{\bs{h}}}}
\end{equation}
is a symmetric matrix. Under this assumption, the functions $(H_1, \ldots, H_n)$ are in involution with each other. Vice versa, if there exists a complete integral of the HJ equations of the form \eqref{eq:CI}, there exist $m$ separation equations of the form \eqref{eq:gSk}.
\end{theorem}
\begin{proof}
Condition \eqref{eq:Jacp} allows us  to solve the GSE \eqref{eq:gSk} with respect to $\boldsymbol{p}$ and to get a system of $n$ partially separated first-order PDEs for the characteristic Hamiltonian function $W(\boldsymbol{q};h_1,\ldots,h_n$), which read
\beq \label{eq:sistW}
\frac{\partial{W}}{\partial{q^{a,j_a}}}=p_{a,j_a}(q^{a,1},\ldots,q^{a,\sigma_{a}};h_1,\ldots,h_n) \qquad a=1,\ldots,v, \quad j_a=1,\ldots,\sigma_a  \ .
\eeq
The integrability of the latter system  is equivalent  to the  compatibility conditions 
\beq \label{eq:compat}
\frac{\partial{p_{b,j_b}}}{\partial{q^{a,j_a}}}=\frac{\partial{p_{a,j_a}}}{\partial{q^{b,j_b}}} \ .
\eeq
In turn, such conditions are equivalent to the symmetry of the matrix product in the l.h.s. of Eq. \eqref{eq:SymPhi}. In fact, from the GSE \eqref{eq:gSk} it follows that for each $k=1,\ldots,n$,
\beq \label{eq:3.7}
0=\rd \phi_{{k}_{|_{{\mathcal{L}_{\bs{h}}}}}}=\sum_{i=1}^n \left(\frac{\partial \phi_k}{\partial q^i}\rd q^i+\frac{\partial \phi_k}{\partial p_i} \rd p_i \right)_{|_{{\mathcal{L}_{\bs{h}}}}} =
\sum_{i=1}^n \bigg(\frac{\partial \phi_k}{\partial q^i}+\sum_{j=1}^n\frac{\partial \phi_k}{\partial p_j} \frac{\partial p_j}{\partial q^i}\bigg)\rd q^i _{|_{{\mathcal{L}_{\bs{h}}}}}\ .
\eeq
Thanks to condition \eqref{eq:Jacp}, Eqs. \eqref{eq:3.7} can be solved for the matrix $\left[\frac{\partial{\boldsymbol{p}}}{\partial{\boldsymbol{q}}}\right]_{\big|_{{\mathcal{L}}_{\bs{h}}}}$, which gives us Eq. \eqref{eq:SymPhi}. Thus,  the functions $(H_1,\ldots, H_n)$ are in involution with each other.  Consequently, $\mathcal{L}_{\bs{h}}$ is a Lagrangian foliation. 

Conversely, if there exists a complete integral of the form \eqref{eq:CI}, then the set of equations 
\beq \label{eq:sistW}
p_{a,j_a}= \frac{\partial{W_a}}{\partial{q^{a,j_a}}}(q^{a,1},\ldots,q^{a,\sigma_{a}};h_1,\ldots,h_n) \qquad a=1,\ldots,v, \quad j_a=1,\ldots,\sigma_a  \eeq
are generalized separation equations of the form \eqref{eq:gSk},  and satisfy the condition ${\rm  det}\left[\frac{\partial \boldsymbol{\phi}}{\partial \boldsymbol{H} }\right] \neq 0$ as a consequence of relation \eqref{eq:CHJ}.

\end{proof}
\begin{remark}
In the case of totally separable systems, the condition of symmetry of matrix \eqref{eq:SymPhi} is obviously fulfilled as both the Jacobian matrices of $\boldsymbol{\phi}$ with respect to $\boldsymbol{q}$ and $\boldsymbol{p}$ are diagonal.
\end{remark}

\subsection{Partial separability and Benenti-type conditions}

We wish to show that the notion of partial separability for completely integrable Hamiltonian systems can be characterized by means of a generalization of Benenti's relations of $\mathcal{T}$-involution \eqref{eq:SI}.

\begin{definition}[$\mathcal{P}$-involution]\label{def:BP} A pair  of functions $(H_\alpha,H_\beta)$ in involution will be said  to be  in partially separated involution (or $\mathcal{P}$-involution) with respect to the set of coordinates $(\boldsymbol{q}^1,\boldsymbol{p}_1,\ldots, \boldsymbol{q}^v,\boldsymbol{p}_v)$, if the reduction of the Poisson brackets to each subset of coordinates 
$(\boldsymbol{q}^a,\boldsymbol{p}_a)$  vanishes:
\begin{equation} \label{eq:BR}
\{H_\alpha,H_\beta\}\lvert_{a}=\sum_{j_a=1}^{ \sigma_a}\left(\frac{\partial H_\alpha}{\partial q^{a,j_a}}
\frac{\partial H_\beta}{\partial p_{a,ja}}-\frac{\partial H_\alpha}{\partial p_{a,ja}}
\frac{\partial H_\beta}{\partial q^{a,ja}}\right)=0 \, , \quad  a=1,\ldots , v \, . 
\end{equation}
\end{definition}
In particular, if each addend in the above sums vanishes, we recover Benenti's case of $\mathcal{T}$-involution.

\vspace{2mm}

The previous definition is motivated by the following
\begin{theorem} \label{theo:5}
Given a Hamiltonian function $H$,  a system of Darboux coordinates $(\boldsymbol{q}^1,\boldsymbol{p}_1,\ldots,\boldsymbol{q}^v, \boldsymbol{p}_v)$ partially separates the complete integral W  \eqref{eq:CI} if and only if there exists a family of independent functions $(H_1, \ldots, H_n)$ which are vertically independent (i.e., they fulfill Eq. \eqref{eq:Hiv}) and in $\mathcal{P}$-involution with each other and with $H$.
\end{theorem}
\begin{proof} 
As is well known, the existence of a complete integral  $W$ of the HJ equation of $H$ is equivalent to the existence of a family of $n$ independent functions $(H_1, \ldots, H_n)$, vertically independent and in involution with each other and with $H$. We will show that $W$ is additively separated if and only if the functions $(H_1,\ldots,H_n)$
are in $\mathcal{P}$-involution. By differentiating Eqs. \eqref{eq:Hi} with respect to the coordinates $(q^1,\ldots,q^n)$ and taking into account Eqs. \eqref{eq:pLF}, one gets 
\begin{equation} \label{eq:dHq}
\frac{\partial{H}_\alpha}{\partial{q_j}}+\sum_{k=1}^n\frac{\partial{H}_\alpha}{\partial{p_k}}\frac{\partial^2 W}{\partial q^j\partial q^k} =0,
\qquad\qquad \alpha,j=1,\ldots,n \ .
\end{equation}
We shall denote by $\Gamma_{jk}=\frac{\partial^2 W}{\partial q^j\partial q^k}$ the elements of the Hessian matrix of $W$.
Let us  write down the \textit{partial} Poisson brackets of two functions $(H_\alpha,H_\beta)$ in matrix form:
\beq \label{eq:3.8}
\{H_\alpha,H_\beta\}\lvert_{a}=
\left\langle\frac{\partial H_\alpha}{\partial \boldsymbol{q}^{a}}, \frac{\partial H_\beta}{\partial\boldsymbol{p}_{a}}\right\rangle
-\left\langle\frac{\partial H_\beta}{\partial\boldsymbol{q}^{a}},\frac{\partial H_\alpha}{\partial \boldsymbol{p}_{a}}\right\rangle=
\left\langle\frac{\partial H_\alpha}{\partial \boldsymbol{q}}~\vline~ \boldsymbol{\Pi}_a ~\vline~
\frac{\partial H_\beta}{\partial\boldsymbol{p}}\right\rangle
-\left\langle\frac{\partial H_\beta}{\partial\boldsymbol{q}}~\vline~ \boldsymbol{\Pi}_a ~\vline \frac{\partial H_\alpha}{\partial \boldsymbol{p}}\right\rangle
\eeq
where  $ \boldsymbol{\Pi}_a$ are the $n\times n$ block-diagonal matrices given by
\beq
 \boldsymbol{\Pi}_a = \text{diag}~[\bs{0},\ldots, \boldsymbol{I}_{\sigma_a} ,\ldots, \bs{0}]
 \eeq
 with $v$ blocks $\sigma_a \times \sigma_a$. 
By solving Eq. \eqref{eq:dHq} for $\frac{\partial H_\alpha}{\partial \boldsymbol{q}^{a}}$ and substituting it into Eq. \eqref{eq:3.8},  thanks to the symmetry of $\boldsymbol{\Gamma}$ and $\boldsymbol{\Pi}_a$, we get 
\begin{equation} \label{eq:PBa}
\{H_\alpha,H_\beta\}\lvert_{a}=
\left\langle-\frac{\partial H_\alpha}{\partial \boldsymbol{p}}\vline ~\boldsymbol{\Gamma} \, \bs{\Pi}_a ~\vline~
\frac{\partial H_\beta}{\partial\boldsymbol{p}}\right\rangle
+\left\langle\frac{\partial H_\beta}{\partial\boldsymbol{p}}\vline~\boldsymbol{\Gamma}\, \boldsymbol{\Pi}_a ~\vline~ \frac{\partial H_\alpha}{\partial \boldsymbol{p}}\right\rangle=-
\left\langle\frac{\partial H_\alpha}{\partial\boldsymbol{p}}\vline~ \Big[\boldsymbol{\Gamma}\ ,\boldsymbol{\Pi}_a \Big ] ~\vline~\frac{\partial H_\beta}{\partial \boldsymbol{p}}\right\rangle
\end{equation}
for $a = 1, \ldots, v$. For each $\alpha,\beta=1,\ldots,n$, we can collect together the partial Poisson brackets  \eqref{eq:PBa} to form an $n\times n$ matrix $\boldsymbol{P}_{a}$. By introducing  the Jacobian matrix $\bs{\mathcal{J}}$ of the functions $(H_1,\ldots,H_n)$ with respect to the momenta $(p_1,\ldots,p_n)$, namely, $[\bs{\mathcal{J}}]_{\alpha k}:=\frac{\partial H_\alpha}{\partial p_k}$,
 we can rewrite Eqs.   \eqref{eq:PBa}  in the following matrix form
 \beq
 \bs{P}_{a}=-\bs{\mathcal{J}} \,\Big[\boldsymbol{\Gamma}\ , \boldsymbol{\Pi}_a \Big ] \,\bs{\mathcal{J}}^T \ .
\eeq

Thus, the condition of $\mathcal{P}$-involution with respect to the set of coordinates $\boldsymbol{q}^a$ is equivalent to the vanishing of the matrix function $ \boldsymbol{P}_{a}$.   In turn, due to the vertical independence \eqref{eq:Hiv} of the integrals of motion $(H_1,\ldots,H_n)$, such a matrix vanishes if and only if $\boldsymbol{\Gamma}$ commutes with $ \boldsymbol{\Pi}_a$, that is 
\begin{equation} \label{eq:comm0}
\Big[\boldsymbol{\Gamma} \ ,  \boldsymbol{\Pi}_a \Big ]=\boldsymbol{0}_n,  \qquad  a = 1, \ldots, v \,\ .
\end{equation}
 Taking into account that: $(i)$ the commutator in the left-hand side of Eq. \eqref{eq:comm0} has possibly non-vanishing elements only in the rows $a, 1; \ldots; a, \sigma_{a}$ and columns $a, 1; \ldots; a, \sigma_{a}$ (each of them containing $n$ elements), with a null $\sigma_a \times \sigma_{a}$ square block defined by the intersection of these rows and columns, and $(ii)$ Eq. \eqref{eq:comm0}  must hold for $a = 1, \ldots, v $, we can conclude that the condition of $\mathcal{P}$-involution is equivalent to the fact that the Hessian matrix  $\boldsymbol{\Gamma}$ has the block-diagonal structure
\beq
\boldsymbol{\Gamma}= \textrm{diag}~[\boldsymbol{\Gamma}_{1},\ldots,\boldsymbol{\Gamma}_{v}] \ ,
\eeq
where $\boldsymbol{\Gamma}_{a}$ are symmetric blocks $\sigma_a \times \sigma_a$, $a = 1, \ldots, v$. In turn, this is equivalent to 
the fact that $W$ is partially separated in $v$ addends, as in Eq. \eqref{eq:CI}. 

Therefore, by substituting $W$ into the HJ equations associated with $H$, we obtain the partial separation of this equation.
\end{proof}

 \section{Non-semisimple Haantjes algebras and partial separability} \label{sec:4}

\subsection{Darboux-Haantjes coordinates and partial separability} A crucial step for the application of our geometric theory of partial separability is the construction of suitable sets of Darboux-Haantjes coordinates. In \cite{TT2022AMPA}, we proved that, given an integrable Hamiltonian system, if we associate it with an Abelian, \textit{semisimple} Haantjes algebra of maximal rank, then we can totally separate the HJ equations in the Darboux-Haantjes coordinates diagonalizing the algebra. The results of this section offer a new, geometric interpretation of the \textit{non-semisimple} scenario (Case \textbf{I} of Definition \ref{def:PS}). Precisely, if an integrable system admits an Abelian, non-semisimple Haantjes algebra, then the DH variables ensuring block-diagonalization of the operators of the algebra also allow us to partially separate the related HJ equations. 

\begin{theorem}[$\mathcal{P}$-involution] \label{th:Nonsemi}
Let $M$ be an Abelian $2n$-dimensional (possibly non-semisimple)
$\omega \mathscr{H}$ manifold of class $m\leq n$ and $( H_{1}, ..., H_{m})$ be a family of $C^{\infty} (M)$ functions belonging to a Haantjes chain generated by a function $H $, via a basis of Haantjes operators $\{ \boldsymbol{K}_{1}, ..., \boldsymbol{K}_{m} \} \in \mathscr{H}$. Then, the $m$ functions  $( H_{1}, ..., H_{m})$ are in $\mathcal{P}$-involution with respect to   any set $(\boldsymbol{q}, \boldsymbol{p})$ of DH coordinates.
\end{theorem}

\begin{proof}

Let $\{ U, (q^{a, j_{a}}, p_{a, j_{a}}) \}, \quad a = 1, ...,v, \quad j_{a} = 1, ..., \sigma_{a} = \frac{1}{2} \, \text{rank} \, (\mathcal{V}_{a})$ be a Darboux-Haantjes chart adapted to the spectral decomposition \eqref{eq:TVa}, where 
\begin{equation}
\mathcal{V}_{a} (\boldsymbol{x}) := \mathcal{D}_{i_{1}}^{(1)} (\boldsymbol{x}) \cap ... \cap \mathcal{D}_{i_{m}}^{(m)} (\boldsymbol{x}) \equiv \text{Span} \Big\{ \dfrac{\partial}{\partial q^{a, j_{a}}}, \dfrac{\partial}{\partial p_{a, j_{a}}} \Big\} .
\end{equation}

We recall that
\begin{equation}
\mathcal{D}_{i_{\alpha}}^{(\alpha)} (\boldsymbol{x}) := \text{ker} (\boldsymbol{K}_{\alpha} - \lambda^{(\alpha)}_{i_{\alpha}} \bs{I})^{\rho_{i_{\alpha}}^{(\alpha)}} (\boldsymbol{x}), \qquad \alpha = 1, ..., m, \quad i_{\alpha} = 1, ..., s_{\alpha},
\end{equation}
where $\rho_{i_{\alpha}}^{(\alpha)}$ is the Riesz index of the eigenvalue $\lambda^{(\alpha)}_{i_{\alpha}}$ and  $s_{\alpha}$ is the number of distinct eigenvalues of $\boldsymbol{K}_{\alpha}$. Due to Eq. \eqref{eq:KNS}, the corresponding chain equations $\rd H_{\alpha} = \boldsymbol{K}_{\alpha}^{T} \rd H$ are
\begin{align}
& \dfrac{\partial H_{\alpha}}{\partial q^{a, j_{a}}} = \sum_{k_{a} = 1}^{\sigma_{a}}
\left( \dfrac{\partial H}{\partial q^{a, k_{a}}} 
 \left( 
  \bs{A}_{a}^{(\alpha)}(\boldsymbol{q}, \boldsymbol{p})
 \right)
 ^{ a, k_{a}}_{ a, j_{a}} 
+
 \dfrac{\partial H}{\partial p_{a, k_{a}}} 
 \left( 
  \bs{C}_{a}^{(\alpha)}(\boldsymbol{q}, \boldsymbol{p})
 \right)
 ^{ a, k_{a}}_{ a, j_{a}} 
 \right),
 \\
&
 \dfrac{\partial H_{\alpha}}{\partial p_{a, j_{a}}} = \sum_{k_{a} = 1}^{\sigma_{a}} 
 \left(\dfrac{\partial H}{\partial q^{a, k_{a}}} \left ( \bs{B}_{a}^{(\alpha)}(\boldsymbol{q}, \boldsymbol{p})\right)^{a, k_{a}}_{a, j_{a}} 
+ 
\dfrac{\partial H}{\partial p_{a, k_{a}}} \left ( (\bs{A}_{a}^{(\alpha)}(\boldsymbol{q}, \boldsymbol{p}))^T\right)^{a, k_{a}}_{a, j_{a}}\right).
\end{align}

Therefore, the partial brackets involved in Benenti's theorem read
\begin{align}\label{eq:PBsingle}
\{ H_{\alpha}, H_{\beta} \} \big|_{{a},j_a}=&
\sum_{i_a,k_{a} = 1}^{\sigma_{a}}
\left( \dfrac{\partial H}{\partial q^{a, i_{a}}} 
 \left( 
  \bs{A}_{a}^{(\alpha)}(\boldsymbol{q}, \boldsymbol{p})
 \right)
 ^{ a, i_{a}}_{ a, j_{a}} 
+
 \dfrac{\partial H}{\partial p_{a, i_{a}}} 
 \left( 
  \bs{C}_{a}^{(\alpha)}(\boldsymbol{q}, \boldsymbol{p})
 \right)
 ^{ a, i_{a}}_{ a, j_{a}} 
 \right)\\ \nn
 &
 \left(\dfrac{\partial H}{\partial q^{a, k_{a}}} \left ( \bs{B}_{a}^{(\beta)}(\boldsymbol{q}, \boldsymbol{p})\right)^{a, k_{a}}_{a, j_{a}} 
+ 
\dfrac{\partial H}{\partial p_{a, k_{a}}} \left ( (\bs{A}_{a}^{(\beta)}(\boldsymbol{q}, \boldsymbol{p}))^T\right)^{a, k_{a}}_{a, j_{a}}
\right)\\ \nn
-
&\sum_{i_a,k_{a} = 1}^{\sigma_{a}}
\left( \dfrac{\partial H}{\partial q^{a, k_{a}}} 
 \left( 
  \bs{B}_{a}^{(\alpha)}(\boldsymbol{q}, \boldsymbol{p})
 \right)
 ^{ a, k_{a}}_{ a, j_{a}} 
+
 \dfrac{\partial H}{\partial p_{a, k_{a}}} 
 \left( 
  (\bs{A}_{a}^{(\alpha)}(\boldsymbol{q}, \boldsymbol{p}))^T
 \right)
 ^{ a, k_{a}}_{ a, j_{a}} 
 \right)\\ \nn
 &
 \left(\dfrac{\partial H}{\partial q^{a, i_{a}}} \left ( \bs{A}_{a}^{(\beta)}(\boldsymbol{q}, \boldsymbol{p})\right)^{a, i_{a}}_{a, j_{a}} 
+ 
\dfrac{\partial H}{\partial p_{a, i_{a}}} \left ( (\bs{C}_{a}^{(\beta)}(\boldsymbol{q}, \boldsymbol{p}))^T\right)^{a, i_{a}}_{a, j_{a}}
\right).
\end{align}
After a long but straightforward calculation (reported in Appendix \ref{ap:A}), we get
\begin{equation} \label{eq:4.6}
\begin{split}
&\{ H_{\alpha}, H_{\beta} \} \big|_{{a}}=
\sum_{i_a,k_a=1}^{\sigma_a }
\bigg(
\dfrac{\partial H}{\partial q^{a, i_a}}  \dfrac{\partial H}{\partial p_{a, k_{a}}} 
\left(
 \left[ \boldsymbol{A}^{(\alpha)}_{a}, \boldsymbol{A}^{(\beta)}_{a}\right ]^{i_a}_{k_a}-
\left( \boldsymbol{B}_{a}^{(\alpha)}(\boldsymbol{C}_{a}^{(\beta)})^T-
\boldsymbol{B}_{a}^{(\beta)}(\boldsymbol{C}_{a}^{(\alpha)})^T\right)^{i_a}_{k_a}
 \right)\\
 &+
\dfrac{\partial H}{\partial q^{a, i_a}}  \dfrac{\partial H}{\partial q^{a, k_{a}}} 
\left(\boldsymbol{A}_{a}^{(\alpha)} (\boldsymbol{B}_{a}^{(\beta)})^T-\boldsymbol{A}_{a}^{(\beta)}(\boldsymbol{B}_{a}^{(\alpha)})^T \right)^{i_a}_{k_a}
+
 \dfrac{\partial H}{\partial p_{a, i_a}}  \dfrac{\partial H}{\partial p_{a, k_{a}}} 
\left(\boldsymbol{C}_{a}^{(\alpha)} \boldsymbol{A}_{a}^{(\beta)}-\boldsymbol{C}_{a}^{(\beta)} \boldsymbol{A}_{a}^{(\alpha)}\right)^{i_a}_{k_a}
\bigg) .
\end{split}
\end{equation}
From the compatibility conditions  \eqref{eq:cOmNS} with $\boldsymbol{\Omega}$,  and   Eqs. \eqref{eq:AAcomm}, it follows that
\[
\dfrac{\partial H}{\partial q^{a, i_a}}  \dfrac{\partial H}{\partial p_{a, k_{a}}} 
\left(
 \left[ \boldsymbol{A}_{a}^{(\alpha)}, \boldsymbol{A}_{a}^{(\beta)}\right ]^{i_a}_{k_a}-
\left( \boldsymbol{B}_{a}^{(\alpha)}(\boldsymbol{C}_{a}^{(\beta)})^T-
\boldsymbol{B}_{a}^{(\beta)}(\boldsymbol{C}_{a}^{(\alpha)})^T\right)^{i_a}_{k_a}
 \right)=0 \, .
 \] 
 Moreover, from Lemma \ref{lm:ABACskew} we deduce that each of the two remaining addends also vanish. This completes the proof. 
\end{proof}
In the case of a semisimple $\omega \mathscr{H}$ manifold, we can deduce a direct consequence of the previous result.
\begin{corollary} 
If the Haantjes algebra $\mathscr{H}$ is semisimple, then the $m$ functions  $(H_{1},\ldots,H_m)$ are in $\mathcal{T}$-involution with respect to  any set $(\boldsymbol{q}, \boldsymbol{p})$ of DH coordinates, that is
\begin{equation} 
\{ H_{\alpha}, H_{\beta} \} \big|_{{a},j_a}=0, \qquad a=1,\ldots,v, \qquad j_a=1,\ldots,\sigma_a \ .
\end{equation}
\begin{proof}
As the algebra $\mathscr{H}$ is semisimple, the blocks $ \boldsymbol{B}_{a}^{(\alpha)}$ and $\boldsymbol{C}_{a}^{(\alpha)}$ vanish whereas the blocks $ \boldsymbol{A}_{a}^{(\alpha)}$ are all diagonal and have the eigenvalues 
$\lambda_{a}^{(\alpha)}$ of $\boldsymbol{K}_\alpha$ as diagonal elements, as can be ascertained from Eq. \eqref{eq:PBsingle}. 
\end{proof}
\end{corollary}
The converse of this corollary  holds as well, as we prove in the following
\begin{theorem} [Jacobi-Haantjes] \label{th:SoVgLc}
Let $(M, \omega)$ be a symplectic manifold and $(H_1,H_2,\ldots,H_m)$ be  $m\leq n$ vertically independent $C^{\infty}(M)$  functions  which are   in $\mathcal{T}$-involution w.r.t. a set of Darboux coordinates $(\boldsymbol{q},\boldsymbol{p})$. Then, these functions belong to the (short) Haantjes  chains generated by the Haantjes operators 
 \begin{equation} \label{eq:KSoV}
\boldsymbol{K}_\alpha=\sum _{i=1}^n \frac{\frac{\partial H_{\alpha}}{\partial p_i}}{ \frac{\partial H}{\partial p_i}}\bigg (\frac{\partial}{\partial q^i}\otimes \rd q^i +\frac{\partial}{\partial p_i}\otimes \rd p_i \bigg ) , \qquad \alpha=1,\ldots,m \, ,
\end{equation}
where $H$ is any function depending on $(H_1, \ldots, H_m)$ and fulfilling $\frac{\partial H}{\partial p_i}\neq 0$, $i=1,\ldots,n$. These operators generate a semisimple $\omega \mathscr{H}$ structure on $M$ of rank $m$.
\end{theorem}
\begin{proof}
The proof proceeds, exactly, as the proof of Theorem 2 of \cite{RTT2022CNS}  by simply interchanging $m$ with $n$.
\end{proof}

\section{Partial separability and generalized St\"ackel systems} \label{sec:5}

The theory of partial separability can also be formulated in the case of Hamiltonian systems admitting a priori less than $n$ integrals of motion  \cite{DPAMPA1896,StAMPA1897,PainlCRAS1897, HavasJMP1976,CR2019SIGMA,DKNSIGMA2019} (Case \textbf{II} of Definition \ref{def:PS}). We shall present here a novel formulation of this theory in the context of Haantjes geometry, which widely extends the approach we presented in \cite{TT2016SIGMA}.  
More precisely, in that article, we have determined the semisimple symplectic-Haantjes manifolds associated with the Arnold-Kozlov-Neishtadt systems \eqref{AKN}. In essence, they are families of $n$ Hamiltonians in involution, which are generated from a standard St\"ackel matrix $\bs{S}$, but keeping arbitrary the dependence of the functions $f_k$ on the momentum $p_k$. By specializing the St\"ackel matrix and vector \eqref{eq:SM} for the AKN systems conveniently, we can also derive directly both families of quasi-bi-Hamiltonian and classical St\"ackel systems. Haantjes geometry allows us to \textit{totally} separate the AKN systems using the construction of related DH coordinates.

We shall now extend  these results to a much more general family of Hamiltonian models, which are \textit{partially separable}.

\subsection{A new family of generalized St\"ackel systems}
In order to introduce the class of systems we are interested, we shall first review the notion of generalized St\"ackel matrix, defined in 1897 by St\"ackel in \cite{StAMPA1897} and by Painlev\'e in \cite{PainlCRAS1897}.

\begin{definition}
Let $(q^1,\ldots,q^n)$ be a coordinate chart, and $\sigma_1,\ldots,\sigma_m\in \mathbb{N}\backslash \{0\}$ be a set of integers such that $\sigma_1+\ldots +\sigma_m=n$. An invertible $m\times m$  matrix of the form
\beq \label{eq:GSM}
\bs{S}=\begin{bmatrix}
S_{11}(\bs{q}^1) & \cdots & S_{1m}(\bs{q}^1) \\
S_{21}(\bs{q}^2) & \cdots & S_{2m}(\bs{q}^2) \\ \vdots & \ddots & \vdots \\ S_{m1}(\bs{q}^m) & \cdots & S_{mm}(\bs{q}^m)
\end{bmatrix}
\eeq
where $\bs{q}^1=(q^{1,1},\ldots,q^{1,\sigma_1}),\ldots,\bs{q}^m=(q^{m,1},\ldots,q^{m,\sigma_m})$, is said to be a generalized St\"ackel matrix.
\end{definition}
Inspired by the previous definition, we propose a novel family of Hamiltonian models.
\begin{definition}
Let $(M, \omega)$ be a symplectic manifold and $(q^i, p_i)$ a set of Darboux coordinates. We introduce the class of generalized St\"ackel Hamiltonians defined by
\begin{equation} \label{eq:GSS}
\begin{bmatrix}
H_1 \\
\vdots \\
H_m
\end{bmatrix}
=\bs{S}^{-1 }
\begin{bmatrix}
f_1 (\boldsymbol{q}^1,\boldsymbol{p}_1)\\
\vdots \\
f_m(\boldsymbol{q}^m,\boldsymbol{p}_m)
\end{bmatrix} \ ,
\end{equation}
where  $(H_1,\ldots,H_m)$  are $m$ functions on  $M$, $\bs{q}^1=(q^{1,1},\ldots,q^{1,\sigma_1})$, $\bs{p}_1=(p_{1,1},\ldots,p_{1,\sigma_1})$, \ldots, $\bs{q}^m=(q^{m,1},\ldots,q^{m,\sigma_m})$, $\bs{p}_m=(p_{m,1},\ldots,p_{m,\sigma_m})$ 
 and $\bs{S}$ is a generalized $m\times m$ St\"ackel   matrix, whose $a$-th row depends  on the variables  $(\bs{q}^a)$ only. Also, 
\beq |\bs{F}\rangle:=
\begin{bmatrix}
f_1 (\boldsymbol{q}^1,\boldsymbol{p}_1)\\
\vdots \\
f_m(\boldsymbol{q}^m,\boldsymbol{p}_m)
\end{bmatrix} 
\eeq
will be said to be a generalized St\"ackel vector.
\end{definition}

\begin{remark}
The family of Hamiltonian systems given by generalized St\"ackel Hamiltonians \eqref{eq:GSS} extend significantly the family of known systems obtainable by the St\"ackel geometry. The AKN systems are realized when $m=n$ (i.e. $\sigma_1=\ldots =\sigma_n=1$), namely when $\bs{S}$ is a standard St\"ackel matrix \eqref{eq:SM}.  In addition, if we impose that the functions $f_k$ depend just on a pair $(q^k,p_k)$ of canonical coordinates and that their dependence on $p_k$ is \textit{quadratic}, then we get back the original St\"ackel systems \cite{StCRAS1893,Pere}. The models introduced by Di Pirro \cite{DPAMPA1896} belong to the class \eqref{eq:GSS} if $m=n-1$ and $f_k$ are quadratic in the momentum $p_k$. The more general systems considered first in \cite{StAMPA1897} and then in \cite{HavasJMP1976}, \cite{CR2019SIGMA} and \cite{DKNSIGMA2019}, related to the Painlev\'e metrics \cite{PainlCRAS1897}, are obtained from the class \eqref{eq:GSS} when $\bs{S}$ is of the form \eqref{eq:GSM} with $m\leq n-1$, and the Hamiltonians $(H_1,\ldots,H_m)$ are still quadratic in the momenta. 
\end{remark}

\subsection{Partial separability and semisimple Haantjes manifolds}

In this section, we will prove the existence of $\omega\mathscr{H}$ manifolds associated with partially separable Hamiltonian systems. We start with the following 
\begin{theorem}[$\mathcal{T}$-Involution] \label{th:Sinvol}
Let $(H_1,\ldots,H_m)$ be $m$ generalized St\"ackel Hamiltonians. 
Then, any pair of functions $(H_a,H_b)$  are in $\mathcal{T}$-involution, that is
\beq
 \{ H_{a}, H_{b} \} \big|_{c, j_{c}}=0, \qquad\qquad c=1,\ldots,m,\quad \  j_c=1,\ldots,\sigma_c \ .
\eeq
\end{theorem}
\begin{proof}
 Observe that $ \frac{\partial f_c}{\partial q^{d,j_d}}=0$ if $c\neq d$. Let us denote by $S^{-1}_{ij}$ the element in the $i$-th row and the $j$-th column of $\boldsymbol{S}^{-1}$. We obtain
\begin{equation} \label{eq:SinvProof}
\begin{split}
 \{ H_{a}, H_{b} \}\big |_{c, j_{c}}=
&\sum_{d=1}^m \bigg(\frac{\partial S^{-1}_{ad}} {\partial q^{c,j_c}}S^{-1}_{bc}-
  \frac{\partial S^{-1}_{bd}}{\partial q^{c,jc}} S^{-1}_{ac}\bigg) f_d \frac{\partial f_c}{\partial p_{c,j_c}}\\
 \end{split} .
\end{equation}
Taking into account the identity
\beq
\partial_x (\bs{S}^{-1})=-\bs{S}^{-1}\partial_x (\bs{S})\bs{S}^{-1} \ ,
\eeq
where $x$ is any variable on which $\bs{S}$ depends, we deduce 
\beq
\frac{\partial S^{-1}_{ad}} {\partial q^{c,j_c}}=-
S^{-1}_{ac}\bigg(\sum_{k=1}^m \frac{\partial S_{ck}}{\partial q^{c,j_c}} S^{-1}_{kd}\bigg)\ .
\qquad
\eeq
This relation implies 
\beq
\frac{\partial S^{-1}_{ad}} {\partial q^{c,j_c}}S_{bc}^{-1}=
-S^{-1}_{ac}
\bigg(\sum_{k=1}^m \frac{\partial S_{ck}}{\partial q^{c,j_c}} S^{-1}_{kd}\bigg) S_{bc}^{-1} \ .
\eeq

It is immediate to observe that the latter identity is symmetric with respect to the interchange of $a$ with $b$. Therefore, each of the terms in the sum of Eq. \eqref{eq:SinvProof} vanishes. 
\end{proof}
\begin{remark}
Interestingly enough, Theorem \ref{th:Sinvol} holds even in the case when each $a$-th row of the St\"{a}ckel matrix $\bs{S}$ may depend also on the variables $(\boldsymbol{q}^a, \boldsymbol{p}_a)$, as well as the $a$-th row of the generalized St\"{a}ckel vector $|\bs{F} \rangle$ (we omit the proof of this observation for simplicity). In this more general scenario,  both the elements $S_{ij}(\boldsymbol{q}^a,\boldsymbol{p}_a)$ and the functions $f_{a}(\boldsymbol{q}^a,\boldsymbol{p}_a)$ will be called \emph{generalized St\"ackel functions} as they are in $\mathcal{P}$-involution and are characteristic functions of the Haantjes web, by analogy with the standard case already discussed in \cite{TT2016SIGMA,TT2022AMPA}.
\end{remark}

We can now prove one of our main results.

\begin{theorem}  \label{Th4}
Let $(M, \omega)$ be a symplectic manifold and $(H_1,H_2,\ldots,H_m)$ be $m$ independent,  $C^{\infty}(M)$ functions of the form \eqref{eq:GSS}. Then, they belong to  the Lenard--Haantjes chain
\begin{equation} \label{eq:HcS}
\boldsymbol{K}_{\alpha} ^T\mathrm{d} H_1=\mathrm{d} H_{\alpha}, \qquad \alpha=1,\ldots,m \ ,
\end{equation}
where $\boldsymbol{K}_{\alpha}$ are   the Haantjes operators defined by

\begin{equation} \label{eq:StackHaan}
\boldsymbol{K}_{\alpha}:=\sum_{i=1}^m \frac{ \tilde{S}_{\alpha i}}{ \tilde{S}_{1i}} \sum_{k=1}^{\sigma_i} \bigg( \frac{\partial}{\partial q^{i,k}}\otimes \mathrm{d} q^{i,k}+ \frac{\partial}{\partial p_{i,k}}\otimes \mathrm{d} p_{i,k}\bigg)
 \qquad\qquad \ \alpha=1,\ldots,m  \ .
\end{equation}
They generate a semisimple $\omega \mathscr{H}$ manifold of class $m$.
\end{theorem}
\begin{proof}

The statement is a consequence of  Theorem \ref{th:SoVgLc}. In particular, the Haantjes operators \eqref{eq:StackHaan} are a realization of the operators \eqref{eq:KSoV} for the case of the generalized St\"ackel systems defined by the Hamiltonians \eqref{eq:GSS}. In fact, it is sufficient to observe that from Eq. \eqref{eq:GSS} we have
\beq
\frac{\frac{\partial H_{\alpha}}{\partial p_{a,k_a}}}{ \frac{\partial H_1}{\partial p_{a,k_a}}}=
\frac{ \tilde{S}_{\alpha a}}{ \tilde{S}_{1a}} \qquad\qquad k_a=1,\ldots, \sigma_a \ .
\eeq
\end{proof}

\subsection{Generalized St\"ackel systems and Hamilton-Jacobi equations}

We now study the problem of partial separability of the Hamilton-Jacobi equations for the generalized St\"ackel systems \eqref{eq:GSS}. The general form of the time-independent Hamilton-Jacobi equation is
\begin{equation}
H \left( \bs{q}^{1}, ..., \bs{q}^{m}, \bs{p}_{1} = \dfrac{\partial W}{\partial \bs{q}^{1}}, ..., \bs{p}_{m} = \dfrac{\partial W}{\partial \bs{q}^{m}} \right) = h,
\end{equation}
where $h$ is an arbitrary constant. Moreover, the $m$ GSE \eqref{eq:Sk} become

\begin{equation} \label{eq:GSE}
\begin{bmatrix}
f_1 (\boldsymbol{q}^1,\boldsymbol{p}_1)\\
\vdots \\
f_m(\boldsymbol{q}^m,\boldsymbol{p}_m)
\end{bmatrix}
-
\boldsymbol{S}
\begin{bmatrix}
H_1 \\
\vdots \\
H_m
\end{bmatrix}
=
\begin{bmatrix}
0 \\
\vdots \\
0
\end{bmatrix} \ .
\end{equation}
They allow to split the HJ equations into the system
\begin{equation} \label{eq:5.10}
f_a\left(\boldsymbol{q}^a, \frac{\partial W} {\partial \boldsymbol{q}^a}\right)-\langle \bs{S}_a~|~\bs{h}~\rangle=0 \qquad \qquad   a=1,\ldots,m \ ,
\end{equation}
where $\bs{S}_a$ is the $a$-th row of the St\"ackel matrix $\bs{S}$, and $\bs{h}=[h_1,\ldots,h_m]^{T}$. Notice that in this context the number of subsets of Darboux coordinates, previously denoted by $v$, coincides with the number $m$ of Hamiltonian functions of the system.

At this stage, one is left with $m$ independent PDEs to solve. This task may be a nontrivial one. However, one can study these PDEs using several techniques, including numerical methods or Lie symmetry analysis, to solve or further reduce them.

\subsubsection{Systems partially separable in (2,1,1) coordinates}
For the sake of concreteness, we illustrate a realization of the class \eqref{eq:GSS}. Let us consider the system

\begin{equation} \label{eq:GSS3}
\begin{bmatrix}
H_1 \\
H_2\\
H_3
\end{bmatrix}
=\bs{S}^{-1 }
\begin{bmatrix}
f_1 (q^{1}, p_{1}, q^{2}, p_{2})\\
f_2(q^{3}, p_{3})\\
f_3(q^{4}, p_{4})
\end{bmatrix}\ ,
\end{equation}
where $f_{1}$, $f_{2}$ and $f_{3}$ are arbitrary functions and the generalized Stäckel matrix takes the form
\begin{equation} \label{eq:5.10}
\bs{S} =
\begin{bmatrix}
S_{11} (q^{1}, q^{2}) & S_{12} (q^{1}, q^{2}) & S_{13} (q^{1}, q^{2}) \\
S_{21} (q^{3}) & S_{22} (q^{3}) & S_{23} (q^{3}) \\
S_{31} (q^{4}) & S_{32} (q^{4}) & S_{33} (q^{4}) 
\end{bmatrix}.
\end{equation}
The functions $S_{ij}$ are arbitrary, except for the condition that $\bs{S}$ must be invertible. Thus, from Eq. \eqref{eq:GSS3} we obtain a huge family of Hamiltonian systems $(H_1, H_2, H_3)$ in involution. It explicitly reads
\begin{eqnarray}
\nn & H_{1} = \dfrac{1}{\det \bs{S}} \left[ (S_{22} S_{33} - S_{23} S_{32}) f_{1} + (S_{13} S_{32} - S_{12} S_{33}) f_{2} + ( S_{12} S_{23} - S_{13} S_{22} ) f_{3} \right], \\ \label{eq:5.17} 
\\
\nn & H_{2} = \dfrac{1}{\det \bs{S}} \left[ (S_{23} S_{31} - S_{21} S_{33}) f_{1} + (S_{11} S_{33} - S_{13} S_{31}) f_{2} + ( S_{13} S_{21} - S_{11} S_{23} ) f_{3} \right], \\ \label{eq:5.18} 
\\
\nn & H_{3} = \dfrac{1}{\det \bs{S}} \left[ (S_{21} S_{32} - S_{22} S_{31}) f_{1} + (S_{12} S_{31} - S_{11} S_{32}) f_{2} + ( S_{11} S_{22} - S_{12} S_{21} ) f_{3} \right]. \\ \label{eq:5.19}
\end{eqnarray}
 Each of these Hamiltonian systems is partially separable in the original coordinates and satisfies Benenti's $\mathcal{T}$-separability conditions. From relations \eqref{eq:5.10} we obtain the  set of three associated HJ equations:  
 \beq
 \begin{cases}
 f_1\left(q^1, \frac{\partial W} {\partial q^1}, q^2, \frac{\partial W} {\partial q^2}\right)- S_{11} h_1-S_{12} h_2-S_{13}h_3=0 \\
 f_2\left(q^3, \frac{\partial W} {\partial q^3}\right)- S_{21} h_1-S_{22}h_2-S_{23}h_3 =0\\
 f_3\left(q^4, \frac{\partial W} {\partial q^4}\right)- S_{31} h_1-S_{32}h_2-S_{33}h_3=0
.
 \end{cases}
 \eeq
 A Hamiltonian system representing a realization of formula \eqref{eq:GSS3} is given in Section \ref{sec:6.2}.
\section{Examples of partially separable Hamiltonian models}\label{sec:6}

As proved in the previous sections, the construction of sets of partially or totally separated  variables for a given Hamiltonian model is directly related to the determination of the $\omega\mathscr{H}$ manifolds associated with the model. To construct these manifolds, one can adopt the following procedure. Assume that we are given $(H_1,\ldots, H_m)$ functions independent and in involution. 

\vspace{2mm}

(1) Within the class of commuting operators algebraically compatible with the symplectic structure, namely fulfilling Eq. \eqref{eq:compOmH}, determine the operators $\bs{K}_{\alpha}$ that solve the chain equations \eqref{eq:MHchain}. 

\vspace{2mm}

(2) Among the solutions found, choose the operators with vanishing Haantjes torsion: $\mathcal{H}_{\bs{K}}(X, Y) = 0$, $\forall$ $X,Y \in \mathfrak{X}(M)$. 

\vspace{2mm}

(3) Out of these Haantjes operators, construct semisimple and/or non-semisimple Abelian $\omega\mathscr{H}$ manifolds. 

\vspace{2mm}

(4) Find the DH coordinates associated with each $\omega\mathscr{H}$ manifold. They will represent either partially or totally separated coordinates for the Hamiltonian system under scrutiny.

\vspace{2mm}

In the following, we shall discuss two examples of the application of this procedure.

\subsection{The Jacobi-Calogero three-particle system} (Case \textbf{I}) \label{sec:6.1}

According to the previous theory, our purpose is to construct both partial and total separation coordinates by applying the Haantjes geometry. As was pointed out in \cite{TT2022AMPA,RTT2022CNS},  given a Hamiltonian system, it might admit several $\omega\mathscr{H}$ manifolds associated. Thus, a given system might be partially separable in some coordinate charts and totally separable in other charts. This fact reveals the intrinsic richness of the theory. 

To illustrate this point and to consider systems falling into Case $\textbf{I}$ of Definition \ref{def:PS}, we shall consider the Jacobi-Calogero model with three particles, with the Hamiltonian function
\begin{equation}
H = \dfrac{1}{2} \left( p_{1}^{2} + p_{2}^{2} + p_{3}^{2} \right) + V_{\text{Cal}},
\end{equation}
where
\begin{equation}
V_{\text{Cal}} = \frac{1}{(q^{1}-q^{2})^2}+\frac{1}{(q^{1}-q^{3})^2}+\frac{1}{(q^{2}-q^{3})^2}.
\end{equation}
This system is both superintegrable and multiseparable, and it was previously discussed in \cite{TT2016SIGMA} from the perspective of the symplectic-Haantjes theory. In particular,  Haantjes algebras associated with the cylindrical, spherical, and parabolic separation coordinates were obtained. Nevertheless, the existence of an $\omega\mathscr{H}$ manifold associated with the algebra of integrals $\{ H, I_{1}, I_{3} \}$, where
\begin{equation}
\begin{aligned}
& I_{1} = p_{1} + p_{2} + p_{3}, \\
& I_{3} = \frac{1}{3} \left( p_{1}^3 + p_{2}^3 + p_{3}^3 \right) + \frac{p_1 + p_2}{(q^{1}-q^{2})^2} + \frac{p_1 + p_3}{(q^{1}-q^{3})^2} + \frac{p_2 + p_3}{(q^{2}-q^{3})^2},
\end{aligned}
\end{equation}
remained as an open problem. We address here and solve this problem by showing that there exists an Abelian, non-semisimple Haantjes algebra associated with the family $\{ H, I_{1}, I_{3} \}$. 

In order to solve the chain equations $\boldsymbol{L}_{i}^{T} \rd H = \rd I_{i}, \ i=1, 3$, we start from the most general operator compatible with the symplectic form:
\begin{equation} \label{eq:Lchart}
\begin{aligned}
& \boldsymbol{L} = \sum_{i, j = 1}^{3} m_{i j} \left( \dfrac{\partial}{\partial q^{i}} \otimes \rd q^{j} + \dfrac{\partial}{\partial p_{j}} \otimes \rd p_{i} \right) + m_{15} \left( \dfrac{\partial}{\partial q^{1}} \otimes \rd p_{2} - \dfrac{\partial}{\partial q^{2}} \otimes \rd p_{1} \right) \\ & \qquad + m_{16} \left( \dfrac{\partial}{\partial q^{1}} \otimes \rd p_{3} - \dfrac{\partial}{\partial q^{3}} \otimes \rd p_{1} \right) + m_{26} \left( \dfrac{\partial}{\partial q^{2}} \otimes \rd p_{3} - \dfrac{\partial}{\partial q^{3}} \otimes \rd p_{2} \right) \\ & \qquad + m_{42} \left( \dfrac{\partial}{\partial p_{1}} \otimes \rd q^{2} - \dfrac{\partial}{\partial p_{2}} \otimes \rd q^{1} \right) + m_{43} \left( \dfrac{\partial}{\partial p_{1}} \otimes \rd q^{3} - \dfrac{\partial}{\partial p_{3}} \otimes \rd q^{1} \right) \\ & \qquad + m_{53} \left( \dfrac{\partial}{\partial p_{2}} \otimes \rd q^{3} - \dfrac{\partial}{\partial p_{3}} \otimes \rd q^{2} \right),
\end{aligned}
\end{equation}
where $m_{11}, \, m_{12}, \, m_{13}, \, m_{21}, \, m_{22}, \, m_{23}, \, m_{31}, \, m_{32}, \, m_{33}, \, m_{15} \, m_{16}, \, m_{26}, \, m_{42}, \, m_{43}, \, m_{53}$ are 15 arbitrary functions that depend on the variables $(q^{1}, q^{2}, q^{3}, p_{1}, p_{2}, p_{3})$.
To make the expressions easier to read, it is useful to represent these operators using the matrix notation
\begin{equation}
\boldsymbol{L} =
\left[\begin{array}{ccc|ccc}
m_{11} & m_{12} & m_{13} & 0 & m_{15} & m_{16} \\
m_{21} & m_{22} & m_{23} & -m_{15} & 0 & m_{26} \\
m_{31} & m_{32} & m_{33} & -m_{16} & -m_{26} & 0 \\ \hline
0 & m_{42} & m_{43} & m_{11} & m_{21} & m_{31} \\
-m_{42} & 0 & m_{53} & m_{12} & m_{22} & m_{32} \\
-m_{43} & -m_{53} & 0 & m_{13} & m_{23} & m_{33} \\
\end{array}\right] .
\end{equation}

Now, we need to fix the arbitrary functions to satisfy the chain equations and ensure a resulting tensor with vanishing Haantjes torsion.

The algebraic chain equation $\boldsymbol{L}_{1}^{T} \rd H = \rd I_{1}$ can be solved by choosing for instance the operator
\begin{equation}
\boldsymbol{L}_{1} = \dfrac{1}{p_{1} + p_{2} + p_{3}}
\left[\begin{array}{ccc|ccc}
1 & 1 & 1 & 0 & 0 & 0 \\
1 & 1 & 1 & 0 & 0 & 0 \\
1 & 1 & 1 & 0 & 0 & 0 \\ \hline
0 & 0 & 0 & 1 & 1 & 1 \\
0 & 0 & 0 & 1 & 1 & 1 \\
0 & 0 & 0 & 1 & 1 & 1 \\
\end{array}\right] ,
\end{equation}
whose Nijenhuis (and consequently, Haantjes) torsion trivially vanishes.
It is semisimple, its minimal polynomial is of second degree,
\begin{equation}
m_{\boldsymbol{L}_{1}} ( \{ \boldsymbol{q},\boldsymbol{p} \}, \lambda) = \lambda \left( \lambda - \dfrac{3}{p_{1} + p_{2} + p_{3}} \right),
\end{equation}
and its generalized eigenvalues and eigendistributions are
\begin{equation}
\begin{aligned}
& \lambda_{1} = \dfrac{3}{p_{1} + p_{2} + p_{3}}, \ \rho_{1} = 1, \qquad \mathcal{D}_{1} = \bigg\langle \dfrac{\partial}{\partial q^{1}} + \dfrac{\partial}{\partial q^{2}} + \dfrac{\partial}{\partial q^{3}}, \dfrac{\partial}{\partial p_{1}} + \dfrac{\partial}{\partial p_{2}} + \dfrac{\partial}{\partial p_{3}} \bigg\rangle, \\
& \lambda_{2} = 0, \ \rho_{2} = 1, \qquad \mathcal{D}_{2} = \bigg\langle \dfrac{\partial}{\partial q^{1}} - \dfrac{\partial}{\partial q^{2}}, \dfrac{\partial}{\partial q^{1}} - \dfrac{\partial}{\partial q^{3}}, \dfrac{\partial}{\partial p_{1}} - \dfrac{\partial}{\partial p_{2}}, \dfrac{\partial}{\partial p_{1}} - \dfrac{\partial}{\partial p_{3}} \bigg\rangle.
\end{aligned}
\end{equation}

Similarly, the algebraic chain equation $\boldsymbol{L}_{3}^{T} \rd H = \rd I_{3}$ can be solved by choosing the Haantjes operator
\begin{equation} \label{eq:6.9}
\boldsymbol{L}_{3} = \dfrac{1}{2 p_{1} - p_{2} - p_{3}}
\left[\begin{array}{ccc|ccc}
 f_{1} & f_{2} & f_{2} & 0 & 0 & 0 \\
 f_{3} & f_{4} & f_{5} & 0 & 0 & 0 \\
 f_{6} & f_{7} & f_{8} & 0 & 0 & 0 \\ \hline
 0 & f_{9} & - f_{9} & f_{1} & f_{3} & f_{6} \\
 - f_{9} & 0 & f_{9} & f_{2} & f_{4} & f_{7} \\
 f_{9} & - f_{9} & 0 & f_{2} & f_{5} & f_{8} \\
\end{array}\right] ,
\end{equation}
where
\begin{align*}
& f_1 = \frac{1}{p_{1}+p_{2}+p_{3}} \left[ 2 p_{1} \left( p_{1}^2 + \frac{1}{(q^{1}-q^{2})^2}+\frac{1}{(q^{1}-q^{3})^2}\right) \right. \\ & \qquad \qquad \left. +(p_{2}+p_{3}) \left( p_{1}^2- p_{2}^2-p_{3}^2 -\frac{2}{(q^{2}-q^{3})^2} \right) \right] , \\
& f_2 = \frac{1}{p_{2}+p_{3}} \left[ \left( p_{1}^2 + \frac{1}{(q^{1}-q^{2})^2} + \frac{1}{(q^{1}-q^{3})^2} \right) (2 p_{1}-p_{2}-p_{3}) - p_{1} \, f_{1} \right], \\
& f_3 = \left( \frac{3 (p_{1} - p_{2})}{2 p_{1} - p_{2} - p_{3}} -\frac{p_{2}}{p_{2} + p_{3}}\right) \,  f_1 + 2 \left(p_{2}^2+\frac{1}{(q^{1}-q^{2})^2}+\frac{1}{(q^{2}-q^{3})^2}\right) \\ & \qquad - \frac{p_{2} + 3 p_{3}}{p_{2}+p_{3}} \left( p_{1}^2 + \frac{1}{(q^{1}-q^{2})^2} + \frac{1}{(q^{1}-q^{3})^2}\right), \\
& f_9 = \frac{2}{p_{2}-p_{3}} \left[ (2 p_{1}-p_{2}-p_{3}) \left( \frac{p_1 + p_{2}}{(q^{1} - q^{2})^3}+\frac{p_1 + p_{3}}{(q^{1} - q^{3} )^3}\right) \right. \\ & \qquad +2 \left(\frac{1}{(q^{1} - q^{2})^3} - \frac{1}{(q^{1}-q^{3})^3} - \frac{2}{(q^{2}-q^{3})^3}\right) \left(p_{2}^2+\frac{1}{(q^{1} - q^{2})^2}+\frac{1}{(q^{2} - q^{3})^2}\right) \\ & \qquad - \left( \frac{2  p_{2}}{2 p_{1} - p_{2} - p_{3}} \left(\frac{1}{(q^{1}-q^{2})^3}-\frac{1}{(q^{1} - q^{3})^3}-\frac{2}{(q^{2} - q^{3})^3}\right) + \frac{1}{(q^{1} - q^{2})^3} + \frac{1}{(q^{1} - q^{3})^3}\right) \, f_{1} \\ & \qquad \left. + \left(\frac{ p_{2}+3 p_{3} }{2 p_{1}-p_{2}-p_{3}} \left(-\frac{1}{(q^{1}- q^{2})^3}+\frac{1}{(q^{1} - q^{3})^3}+\frac{2}{(q^{2} - q^{3})^3}\right) +\frac{2}{(q^{1}- q^{3} )^3}+\frac{2}{(q^{2} - q^{3})^3}\right) f_{2}   \right],
\end{align*}
and
\begin{equation}
\begin{aligned}
& f_{4} = f_{1} + \dfrac{1}{2} \left( f_{2} - f_{3} \right), \\
& f_{5} = \dfrac{1}{2} \left( 3 f_{2} - f_{3} \right), \\
& f_{6} = 2 f_{2} - f_{3}, \\
& f_{7} = \dfrac{1}{2} \left( f_{2} + f_{3} \right), \\
& f_{8} = f_{1} - \dfrac{1}{2} \left( f_{2} - f_{3} \right).
\end{aligned}
\end{equation}

The operator \eqref{eq:6.9} is non-semisimple. Its minimal polynomial is of third degree and reads
\begin{equation}
m_{\boldsymbol{L}_{3}} ( \{ \boldsymbol{q},\boldsymbol{p} \}, \lambda) = \left( \lambda - \dfrac{f_1 + 2 f_2}{2 p_{1} - p_{2} - p_{3}} \right) \left( \lambda - \dfrac{f_1 - f_2}{2 p_{1} - p_{2} - p_{3}} \right)^{2}.
\end{equation}
Its generalized eigenvalues and eigendistributions are
\begin{equation}
\begin{aligned}
& \lambda_{3} = \dfrac{f_1 + 2 f_2}{2 p_{1} - p_{2} - p_{3}}, \ \rho_{3} = 1, \\ & \qquad \qquad \mathcal{D}_{3} = \bigg\langle \dfrac{\partial}{\partial q^{1}} + \dfrac{\partial}{\partial q^{2}} + \dfrac{\partial}{\partial q^{3}}, \dfrac{\partial}{\partial p_{1}} + \dfrac{\partial}{\partial p_{2}} + \dfrac{\partial}{\partial p_{3}} \bigg\rangle \equiv \mathcal{D}_{1}, \\
& \lambda_{4} = \dfrac{f_1 - f_2}{2 p_{1} - p_{2} - p_{3}}, \ \rho_4 = 2, \\ & \qquad \qquad \mathcal{D}_{4} = \bigg\langle \dfrac{\partial}{\partial q^{1}} - \dfrac{\partial}{\partial q^{2}}, \dfrac{\partial}{\partial q^{1}} - \dfrac{\partial}{\partial q^{3}}, \dfrac{\partial}{\partial p_{1}} - \dfrac{\partial}{\partial p_{2}}, \dfrac{\partial}{\partial p_{1}} - \dfrac{\partial}{\partial p_{3}} \bigg\rangle \equiv \mathcal{D}_{2}.
\end{aligned}
\end{equation}

The algebra $\mathscr{H}_1:=Span\{ \boldsymbol{I}, \boldsymbol{L}_{1}, \boldsymbol{L}_{3} \}$ is a non-semisimple Abelian algebra of Haantjes operators of rank 3; thus, they can be simultaneously block-diagonalized in a DH chart.

As the generalized eigen-distributions coincide, their intersections are trivial and do not provide a finer block-diagonalization for the operators of this algebra. The corresponding annihilators are:
\begin{equation}
\begin{aligned}
& (\mathcal{D}_{1})^{\circ} = \langle \rd q^{1} - \rd q^{2}, \rd q^{1} - \rd q^{3}, \rd p_{1} - \rd p_{2}, \rd p_{1} - \rd p_{3} \rangle, \\
& (\mathcal{D}_{2})^{\circ} = \langle \rd q^{1} + \rd q^{2} + \rd q^{3}, \rd p_{1} + \rd p_{2} + \rd p_{3} \rangle.
\end{aligned}
\end{equation}
By integrating them, we can construct a set of Haantjes coordinates. Among them, we can select those  satisfying the canonical commutation relations, obtaining the following DH coordinates:
\begin{equation} \label{eq:Qpsep}
\begin{aligned}
& Q^{1} = \dfrac{1}{6} \left( 2 \, q^{1} - q^{2} - q^{3} \right) \, , \\
& P_{1} = 2 \, p_{1} - p_{2} - p_{3} \, , \\
\end{aligned} \qquad
\begin{aligned}
& Q^{2} = \dfrac{1}{2} \left( q^{2} - q^{3} \right) \, , \\
& P_{2} = p_{2} - p_{3} \, , \\
\end{aligned} \qquad
\begin{aligned}
& Q^{3} = \dfrac{1}{3} \left( q^{1} + q^{2} + q^{3} \right) \, , \\
& P_{3} = p_{1} + p_{2} + p_{3} \, .
\end{aligned}
\end{equation}
Notice that these simple relations become non-trivial if we express the relation between partial and total separation variables. More precisely, the transformation relating partial and total separation coordinates $(\mathfrak{q}^1, \mathfrak{q}^2, \mathfrak{q}^3)$ is not of the form $Q^1=F_1(\mathfrak{q}^1,\mathfrak{q}^2)$, $Q^2=F_2(\mathfrak{q}^1,\mathfrak{q}^2)$, $Q^3=F_3(\mathfrak{q}^3)$, which a priori guarantees partial separability (with $F_1$, $F_2$, $F_3$ arbitrary, invertible functions). Indeed, once written, for instance, in terms of spherical coordinates, the partial separation ones read
\begin{equation}
\begin{aligned}
& Q^{1} = -\frac{1}{6} \, r \, (\sin \theta \, (\sin \varphi -2 \cos \varphi )+\cos \theta ), \\
& Q^{2} = \frac{1}{2} \, r \, ( \sin \theta \sin \varphi - \cos \theta ), \\
& Q^{3} = \frac{1}{3} \, r \, (\sin \theta \, (\sin \varphi +\cos \varphi )+\cos \theta ).
\end{aligned}
\end{equation}
Similar expressions hold for the cases of cylindrical and parabolic coordinates.

\begin{remark} There is much freedom in the representation of the foliation obtained by integrating the distribution $(\mathcal{D}_{1})^{\circ}$. Indeed, more generally, any choice of a transformation of coordinates of the form
\begin{equation}
Q^1 = \mathfrak{f}_1 (q^{1} - q^{2}, q^{1} - q^{3}) \, , \quad Q^2 = \mathfrak{f}_2 (q^{1} - q^{2}, q^{1} - q^{3}) \, , \quad Q^3 = \mathfrak{f}_3(q^{1} + q^{2} + q^{3}) \, , 
\end{equation}
would provide Haantjes coordinates for the algebra $\mathscr{H}_1$, where $\mathfrak{f}_1$, $\mathfrak{f}_2$ and $\mathfrak{f}_3$ are arbitrary functions that make the transformation invertible.
\end{remark}

In the chart \eqref{eq:Qpsep}, the Hamiltonian function and the integrals of the motion read
\begin{equation}
\begin{aligned}
& H = \frac{1}{12} \left( P_1^2 + 3 P_2^2 + 2 P_{3}^2 \right) +\frac{1}{4 (Q^{2})^2}+\frac{1}{(3 Q^{1}+Q^{2})^2}+\frac{1}{(3 Q^{1}-Q^{2})^2}, \\
& I_1 = P_{3}, \\
& I_3 = \frac{1}{108} \left[ P_1^3 +4 P_{3}^3 + 6 P_1^2 P_{3} - 9 P_1 P_2^2 +18 P_2^2 P_{3} -\frac{9 (P_1-2 P_{3})}{(Q^{2})^2} \right. \\ & \qquad \left.  + 18 ( P_1 + 4 P_{3} ) \left( \frac{1}{(3 Q^{1}+Q^{2})^2} + \frac{1}{(3 Q^{1}-Q^{2})^2} \right) - 54 P_2 \left( \frac{1}{(3 Q^{1}+Q^{2})^2} - \frac{1}{(3 Q^{1}-Q^{2})^2} \right) \right];
\end{aligned}
\end{equation}
and the operators $\boldsymbol{L}_{1}$ and $\boldsymbol{L}_{3}$ in the coordinates $(Q^{1}, Q^{2}, P_{1}, P_{2}, Q^{3}, P_{3})$ take the \textit{block-diagonal form}
\begin{equation}
\boldsymbol{L}_{1} = \dfrac{3}{P_{3}}
\left[\begin{array}{cccc|cc}
0 & 0 & 0 & 0 & 0 & 0 \\
0 & 0 & 0 & 0 & 0 & 0 \\
0 & 0 & 0 & 0 & 0 & 0 \\ 
0 & 0 & 0 & 0 & 0 & 0 \\ \hline
0 & 0 & 0 & 0 & 1 & 0 \\
0 & 0 & 0 & 0 & 0 & 1 \\
\end{array}\right]
\end{equation}
and
\begin{equation}
\boldsymbol{L}_{3} = \frac{1}{2 (Q^{2})^2 \left(9 (Q^{1})^2-(Q^{2})^2\right)^2}
\left[\begin{array}{cccc|cc}
g_1 & 0 & 0 & 0 & 0 & 0 \\
g_2 & g_1 & 0 & 0 & 0 & 0 \\
0 & g_3 & g_1 & g_2 & 0 & 0 \\ 
-g_3 & 0 & 0 & g_1 & 0 & 0 \\ \hline
0 & 0 & 0 & 0 & g_4 & 0 \\
0 & 0 & 0 & 0 & 0 & g_4 \\
\end{array}\right],
\end{equation}
where
\begin{align*}
& g_1 = \frac{1}{3 P_1} \left[  (Q^{2})^2 (9 (Q^{1})^2 - (Q^{2})^2)^2 (P_1^2 - 3 P_2^2 +  4 P_1 P_3) + 9( - 27 (Q^{1})^4 + 18 (Q^{1})^2 (Q^{2})^2 + (Q^{2})^4 ) \right] , \\
& g_2 = \frac{3}{P_1^2} \left[ (Q^{2})^2 (9 (Q^{1})^2 - (Q^{2})^2)^2 (- P_1^2 + P_{2}^2) P_{2} + 24 Q^{1} (Q^{2})^3 P_1 + 3 (27 (Q^{1})^4 - 18 (Q^{1})^2 (Q^{2})^2 - (Q^{2})^4) P_{2} \right] , \\
& g_3 = -\frac{3}{ (Q^{2})^3 \left(9 (Q^{1})^2 - (Q^{2})^2 \right)^3
   P_1^2} \left[ 72 Q^{1} (Q^{2})^5 (9 (Q^{1})^2 - (Q^{2})^2)^2 (3 (Q^{1})^2 + (Q^{2})^2) P_1 P_{2} \right. \\ & \qquad - (Q^{2})^{10} (27 + (Q^{2})^2 (P_1^2 - 9 P_{2}^2)) - 162 (Q^{1})^4 (Q^{2})^6 (63 + 5 (Q^{2})^2 (P_1^2 + 3 P_{2}^2)) \\ & \qquad + 9 (Q^{1})^2 (Q^{2})^8 (-117 + (Q^{2})^2 (5 P_1^2 + 3 P_{2}^2)) + 
 1458 (Q^{1})^6 (Q^{2})^4 (3 + (Q^{2})^2 (5 P_1^2 + 7 P_{2}^2)) \\ & \qquad \left. + 
 6561 (Q^{1})^8 (9 (Q^{2})^2 + 5 (Q^{2})^4 (-P_1^2 + P_{2}^2) + 
    9 (Q^{1})^2 (-1 + (Q^{2})^2 (P_1^2 - P_{2}^2))) \right], \\
& g_4 = \frac{1}{3 P_3} \left[ 27 (3 (Q^{1})^2 + (Q^{2})^2)^2 + (Q^{2})^2 ( 9 (Q^{1})^2 - (Q^{2})^2)^2 (P_1^2 + 3 P_{2}^2 + 2 P_3^2) \right].
\end{align*}

The  Hamilton-Jacobi equations associated with this system read
\begin{equation}
\begin{aligned}
& \left( \dfrac{\partial W_1}{\partial Q^{1}} \right)^2+3 \left( \dfrac{\partial W_1}{\partial Q^{2}} \right)^2 +\frac{3}{(Q^{2})^2}+\frac{12}{(3 Q^{1}+Q^{2})^2}+\frac{12}{(3 Q^{1}-Q^{2})^2} = h, \\ \\
& \dfrac{\rd W_2}{\rd Q^{3}} = h_1, \\ \\
& \dfrac{\partial W_1}{\partial Q^{1}} \left[ \left( \dfrac{\partial W_1}{\partial Q^{2}} \right)^2 +\frac{1}{(Q^{2})^2} - \frac{1}{2} \, \frac{1}{(3 Q^{1}+Q^{2})^2} - \frac{1}{2} \, \frac{1}{(3 Q^{1}-Q^{2})^2} - h + \frac{1}{6} \, h_1^2 \right] \\ & \qquad - \frac{54 \, Q^{1} Q^{2}}{\left( 9 (Q^{1})^2 - (Q^{2})^2\right)^2} \dfrac{\partial W_1}{\partial Q^{2}} = h_{2},
\end{aligned}
\end{equation}
where $h$, $h_1$, $h_2$ are real constants.

\subsection{Construction of a new generalized Stäckel system} \label{sec:6.2} (Case \textbf{II})

We can construct directly new families of partially separable systems by specializing the arbitrary functions involved in definitions \eqref{eq:GSM} and \eqref{eq:GSS}. For instance, working in an eight-dimensional phase space, we can specialize directly Eqs. \eqref{eq:GSS3} and \eqref{eq:5.10} by making the choice 
\begin{equation}
\bs{S} =
\begin{bmatrix}
q^{1} q^{2} & 0 & 0 \\
0 & 1 & q^{3} \\
q^{4} & 0 & 1
\end{bmatrix}
\end{equation}
for the generalized St\"ackel matrix and
\begin{align}
& \nn f_{1} = q^{1} q^{2} (p_{1}^{2} + p_{2}^{2} + q^{1} q^{2}) , \\
& f_{2} = p_{3}^{2} + q^{3} , \\
& \nn f_{3} = p_{4}^{2} + q^{4} ,
\end{align}
for the elements of the St\"{a}ckel vector. In this way, we obtain the Hamiltonian system
\begin{align}
& \nn H_{1} = p_{1}^{2} + p_{2}^{2} + q^{1} q^{2} , \\
& H_{2} = q^{3} q^{4} (p_{1}^{2} + p_{2}^{2}) + p_{3}^{2} - q^{3} p_{4}^{2} + q^{3} q^{4} (q^{1} q^{2} - 1) + q^{3} , \\
& \nn  H_{3} = -q^{4} (p_{1}^{2} + p_{2}^{2}) + p_{4}^{2} + q^{4} (1 - q^{1} q^{2}) .
\end{align}

The three GSE \eqref{eq:GSE} explicitly read
\begin{align}
\begin{cases}
 q^1q^2H_1-q^1q^2(p_1^2 +p_2^2+q^1q^2)=0 , \\
   H_2+q^3 H_3-p_3^2-q^3=0 , \\
   q^4H_1+H_3-p_4^2-q^4=0 .
\end{cases}
\end{align}

Therefore, assuming that
\begin{equation}
W (q^{1}, q^{2}, q^{3}, q^{4}; h_1,h_2,h_3) = W_{1} (q^{1}, q^{2};h_1,h_2,h_3) + W_{2} (q^{3};h_1,h_2,h_3) + W_{3} (q^{4};h_1,h_2,h_3),
\end{equation}
 we obtain the partially separated Hamilton-Jacobi equations
\begin{align}
\begin{cases}
  \left( \dfrac{\partial W_{1}}{\partial q^{1}} \right)^{2} + \left( \dfrac{\partial W_{1}}{\partial q^{2}} \right)^{2}+q^1q^2 =h_1, \\
    \left( \dfrac{\partial W_{2}}{\partial q^{3}} \right)^{2} +q^3-h_3 q^3=h_2, \\
    \left( \dfrac{\partial W_{3}}{\partial q^{4}} \right)^{2} +q^4-h_1 q^4 =h_3 .
\end{cases}
\end{align}

In the case under scrutiny, we can determine an Abelian and semisimple Haantjes algebra $ \mathscr{H}_2 := Span\{ \boldsymbol{K}_{1},\boldsymbol{K}_{2} =\boldsymbol{I}, \boldsymbol{K}_{3} \}$, allowing for partial separation.  Applying Theorem \ref{Th4} by choosing $H=H_{2}$ to be the generator function of the Haantjes chains,  we obtain in the chart $(q^{1}, q^{2}, q^{3}, q^{4}, p_1, p_2, p_3, p_4)$ that the Haantjes operators
\beq
\boldsymbol{K}_{1} = \dfrac{1}{q^{3} q^{4}}\,\textrm{diag} [1,1,0,0,1,1,0,0], \quad \boldsymbol{K}_{3} = - \dfrac{1}{q^{3}}\,\textrm{diag} [1,1,0,1,1,1,0,1]
\eeq
satisfy the chain equations $\rd H_{1} = \boldsymbol{K}_{1}^{T} \rd H$ and $\rd H_{3} = \boldsymbol{K}_{3}^{T} \rd H$, respectively.

\begin{remark}
In \cite{TT2016SIGMA}, it has been shown that Killing tensors for the AKN systems \cite{AKN1997} can be naturally constructed within the theory of $\omega \mathscr{H}$ manifolds by projecting the Haantjes operators associated with the AKN systems from phase space into the configuration space. This, in turn, provides a direct connection between the Haantjes geometry and Eisenhart's theory of Killing tensors for St\"ackel systems \cite{EisenAM1934}.

In our opinion, the main advantage of Haantjes geometry is that it allows us to study systematically the problem of separability within a unified theoretical framework,  valid also for systems outside the classical Stäckel geometry.
\end{remark}

\section{Open problems and future perspectives}\label{sec:7}

Many interesting problems concerning the present theory are open and left for future investigation. 

A first, relevant problem is to formulate Theorem \ref{th:IPS} in the case when the vector-valued function $\bs{\phi}$ of Definition \ref{def:PS} is affine in $(H_1,\ldots,H_n)$, with the aim of possibly associating with these Hamiltonians a generalized St\"ackel matrix.

We believe that the theory of $\mathcal{P}$-involution for completely integrable systems (Case $\textbf{I}$) can be useful to study Hamiltonian systems non-totally separable. For instance, the Calogero model with four particles is only partially separable \cite{WW2005JNMP}. Other models, usually considered non-separable, could also be investigated from this perspective.

\vspace{2mm}

We observe that the class of generalized St\"ackel systems \eqref{eq:GSS}, which a priori possess less than $n$ integrals of motion in involution, may admit a subfamily of completely integrable systems. Thus, it would be important to ascertain in full generality which conditions on the generalized St\"ackel matrix \eqref{eq:GSM} ensure the existence of further integrals necessary to complete the set of available Hamiltonians in involution. 

\vspace{2mm}

Another relevant question concerns the existence of non-semisimple Haantjes algebras. Given an integrable system, we have proved in \cite{TT2022AMPA} that one can construct an associated semisimple Haantjes algebra due to the existence of action-angle variables defined on the Lagrangian foliation of invariant Liouville-Arnold tori. When we deal with superintegrable systems, additional integrals of motion are available. They are in involution with the Hamiltonian, but not necessarily with each other. Presently, it is an open problem whether one might associate with these additional integrals  symplectic Haantjes manifolds, in particular non-semisimple ones. This is the case, for instance, for the Jacobi-Calogero three-particle system discussed in Section \ref{sec:6.1}. We believe that the solution to this problem is strictly related to the general study of normal forms for Haantjes operators. Indeed, unlike the semisimple case, where there exists a unique normal form (the diagonal one), in the non-semisimple case much less is known.

To conclude, it would be interesting to study the co-isotropic foliation associated with partially separable systems in the case of an incomplete family of Hamiltonians (Case \textbf{II}).

All these questions are currently under investigation.

\par

\vspace{3mm}

\textbf{Data availability statement}. 

This manuscript has no associated data.

\vspace{3mm}

\textbf{Conflict of interest statement}. 

The authors declare that they have no conflict of interests.

\section*{Acknowledgement}

We wish to thank gratefully the anonymous Referee for many useful comments, which improved the quality of the paper.

The research of P. T. has been supported by the Severo Ochoa Programme for Centres of Excellence in R\&D
(CEX2019-000904-S), Ministerio de Ciencia, Innovaci\'{o}n y Universidades y Agencia Estatal de Investigaci\'on, Spain.

The research of G. T. has been supported by the research project FRA2022-2023, Universit\'a degli Studi di Trieste, Italy, and by the MMNLP Project - CSN4 of INFN, Italy. 

D. R. N. acknowledges the financial support of EXINA S.L.

P. T. is a member of the Gruppo Nazionale di Fisica Matematica (GNFM) of INDAM. 

\appendix 
\section{Partial brackets in the proof of Theorem \ref{th:Nonsemi}.} \label{ap:A}

We report here the explicit computation leading to Eq. \eqref{eq:4.6}:
\begin{equation}
\begin{split}
&\{ H_{\alpha}, H_{\beta} \} \big|_{a} =
\sum_{j_a=1}^{\sigma_a }\bigg(
\dfrac{\partial H_\alpha}{\partial q^{a, j_a}}  \dfrac{\partial H_\beta}{\partial p_{a, j_{a}}} -
\dfrac{\partial H_\beta}{\partial q^{a, j_a}}  \dfrac{\partial H_\alpha}{\partial p_{a, j_{a}}} 
\bigg )\\
&=\sum_{j_a=1}^{\sigma_a }\Bigg(\bigg(\sum_{i_a=1}^{\sigma_a}
\dfrac{\partial H}{\partial q^{a, i_a}} (\bs{A}_{a}^{(\alpha)})^{i_a}_{j_a} +\dfrac{\partial H}{\partial p_{a, i_a}} (\bs{C}_{a}^{(\alpha)})^{i_a}_{j_a} \bigg)
\bigg(\sum_{k_a=1}^{\sigma_a}
\dfrac{\partial H}{\partial q^{a, k_a}} (\bs{B}_{a}^{(\beta)})^{k_a}_{j_a} +\dfrac{\partial H}{\partial p_{a, k_a}}(( {\bs{A}_{a}^{(\beta)})^{T})}^{k_a}_{j_a} 
\bigg) \\
&-
\sum_{k_a=1}^{\sigma_a} \bigg(
\dfrac{\partial H}{\partial q^{a, k_a}} (\bs{B}_{a}^{(\alpha)})^{k_a}_{j_a} +\dfrac{\partial H}{\partial p_{a, k_a}}( (\bs{A}_{a}^{(\alpha)})^T)^{k_a}_{j_a} \bigg)
\bigg(\sum_{i_a=1}^{\sigma_a}
\dfrac{\partial H}{\partial q^{a, k_a}} (\bs{A}_{a}^{(\beta)})^{i_a}_{j_a} +\dfrac{\partial H}{\partial p_{a, k_a}}( (\bs{C}_{a}^{(\beta)})^T)^{i_a}_{j_a} 
\bigg)\Bigg)
\end{split}
\end{equation}
\begin{equation*}
\begin{split}
&=
\sum_{i_a,j_a,k_a=1}^{\sigma_a} \Bigg(
\dfrac{\partial H}{\partial q^{a, i_a}}\dfrac{\partial H}{\partial p_{a, k_a}} 
\big((\bs{A}_{a}^{(\alpha)})^{i_a}_{j_a}(\bs{A}_{a}^{(\beta)})^{j_a}_{k_a}-(\bs{A}_{a}^{(\alpha)})^{j_a}_{k_a}(\bs{A}_{a}^{(\beta)})^{i_a}_{j_a}
\big)
\\
&+
\dfrac{\partial H}{\partial q^{a, k_a}}\dfrac{\partial H}{\partial p_{a, i_a}} 
\big((\bs{C}_{a}^{(\alpha)})^{i_a}_{j_a}(\bs{B}_{a}^{(\beta)})^{k_a}_{j_a}-(\bs{B}_{a}^{(\alpha)})^{k_a}_{j_a}(\bs{C}_{a}^{(\beta)})^{i_a}_{j_a}
\big)
\\
&+
\dfrac{\partial H}{\partial q^{a, i_a}}\dfrac{\partial H}{\partial q^{a, k_a}} 
\big((\bs{A}_{a}^{(\alpha)})^{i_a}_{j_a}(\bs{B}_{a}^{(\beta)})^{k_a}_{j_a}-(\bs{B}_{a}^{(\alpha)})^{k_a}_{j_a}(\bs{A}_{a}^{(\beta)})^{i_a}_{j_a}
\big)
\\
&+
\dfrac{\partial H}{\partial p_{a, i_a}}\dfrac{\partial H}{\partial p_{a, k_a}} 
\big((\bs{C}_{a}^{(\alpha)})^{i_a}_{j_a}(\bs{A}_{a}^{(\beta)})^{j_a}_{k_a}-(\bs{A}_{a}^{(\alpha)})^{j_a}_{k_a}(\bs{C}_{a}^{(\beta)})^{i_a}_{j_a}
\big)
\Bigg) \ .
\end{split}
\end{equation*}

\end{document}